\documentclass[11pt,letter]{article}
\usepackage[margin=1in]{geometry}
\usepackage[utf8x]{inputenc}
\usepackage{amsmath, amsthm}
\usepackage{wrapfig,floatflt,graphicx,amssymb,textcomp,array,amsmath}
\usepackage{algpseudocode} 
\usepackage{enumerate}
\usepackage{multirow}
\usepackage{tabularx}
\usepackage{algorithm}
\usepackage{color}
\usepackage{todonotes}
\usepackage[titletoc,title]{appendix}

\newcommand{\etal}{{et~al. }}

\newcommand{\removed}[1]{}

\title{On the Minimum Consistent Subset Problem
}

\author{Ahmad Biniaz\thanks{Cheriton School of Computer Science, University of Waterloo, ahmad.biniaz@gmail.com}
	\and Sergio Cabello\thanks{Department of Mathematics, IMFM and FMF, University of Ljubljana, sergio.cabello@fmf.uni-lj.si} 
	\and Paz Carmi\thanks{Department of Computer Science, Ben-Gurion University of the Negev, carmip@cs.bgu.ac.il} 
	\and Jean-Lou De Carufel\thanks{School of Electrical Engineering and Computer Science, University of Ottawa, jdecaruf@uottawa.ca} 
	\and Anil Maheshwari\thanks{School of Computer Science, Carleton University, \{anil, michiel\}@scs.carleton.ca, saeed.mehrabi@carleton.ca}
	\and  Saeed Mehrabi\footnotemark[5]
	\and  Michiel Smid\footnotemark[5]
}

\date{}
\newtheorem{lemma}{Lemma}
\newtheorem{corollary}{Corollary}

\newtheorem{theorem}{Theorem}

\newtheorem*{problem*}{Problem}

\newtheorem*{invariant*}{Invariant}
\newtheorem{question}{Question}

\newcommand{\RR}{\mathbb{R}}

\newcommand{\CH}[1]{CH(#1)}
\newcommand{\bisector}[2]{\beta(#1,#2)}

\begin{document}
	\maketitle
	\begin{abstract}
	Let $P$ be a set of $n$ colored points in the plane. Introduced by Hart (1968), a {\em consistent subset} of $P$, is a set $S\subseteq P$ such that for every point $p$ in $P\setminus S$, the closest point of $p$ in $S$ has the same color as $p$. The consistent subset problem is to find a consistent subset of $P$ with minimum cardinality. This problem is known to be NP-complete even for two-colored point sets. Since the initial presentation of this problem, aside from the hardness results, there has not been significant progress from the algorithmic point of view. In this paper we present the following algorithmic results: 
	\begin{enumerate}
		\item The first subexponential-time algorithm for the consistent subset problem.
		\item An $O(n\log n)$-time algorithm that finds a consistent subset of size two in two-colored point sets (if such a subset exists). Towards our proof of this running time we present a deterministic $O(n \log n)$-time algorithm for computing a variant of the compact Voronoi diagram; this improves the previously claimed expected running time.
		\item An $O(n\log^2 n)$-time algorithm that finds a minimum consistent subset in two-colored point sets where one color class contains exactly one point; this improves the previous best known $O(n^2)$ running time which is due to Wilfong (SoCG 1991).
		\item An $O(n)$-time algorithm for the consistent subset problem on collinear points; this improves the previous best known $O(n^2)$ running time. 
		\item A non-trivial $O(n^6)$-time dynamic programming algorithm for the consistent subset problem on points arranged on two parallel lines.
	\end{enumerate}
	
	To obtain these results, we combine tools from planar separators, paraboloid lifting, additively-weighted Voronoi diagrams with respect to convex distance functions, point location in farthest-point Voronoi diagrams, range trees, minimum covering of a circle with arcs, and several geometric transformations.
	\end{abstract}
	
	\section{Introduction}
	One of the important problems in pattern recognition is to classify new objects according to the current objects using the nearest neighbor rule. Motivated by this problem, in 1968, Hart \cite{Hart1968} introduced the notion of {\em consistent subset} as follows. For a set $P$ of colored points\footnote{In some previous works the points have labels, as opposed to colors.} in the plane, a set $S\subseteq P$ is a consistent subset if for every point $p\in P\setminus S$, the closest point of $p$ in $S$ has the same color as $p$. The {\em consistent subset problem} asks for a consistent subset with minimum cardinality.
	Formally, we are given a set $P$ of $n$ points in the plane that is partitioned into ${P_1,\dots,P_k}$,
	with $k \geqslant 2$, and the goal is to find an smallest set $S\subseteq P$ such that for every $i\in\{1,\dots,k\}$ it holds that if $p\in P_i$ then the nearest neighbor of $p$ in $S$ belongs to $P_i$. It is implied by the definition that $S$ should contain at least one point from every $P_i$. 	
	To keep the terminology consistent with some recent works on this problem we will be dealing with colored points instead of partitions, that is, we assume that the points of $P_i$ are colored $i$. Following this terminology, the consistent subset problem asks for a smallest subset $S$ of $P$ such that the color of every point $p\in P\setminus S$ is the same as the color of its closest point in $S$. The notion of consistent subset has a close relation with Voronoi diagrams, a well-known structure in computational geometry. Consider the Voronoi diagram of a subset $S$ of $P$. Then, $S$ is a consistent subset of $P$ if and only if for every point $s\in S$ it holds that the points of $P$, that lie in the Voronoi cell of $s$, have the same color as $s$; see Figure~\ref{Voronoi-fig}(a).

	Since the initial presentation of this problem in 1968, there has not been significant progress from the algorithmic point of view. Although there were several attempts for developing algorithms, they either did not guarantee the optimality \cite{Gates1972, Hart1968, Wilfong1992} or had exponential running time \cite{Ritter1975}.
	In SoCG 1991, Wilfong \cite{Wilfong1992} proved that the consistent subset problem is NP-complete if the input points are colored by at least three colors---the proof is based on the NP completeness of the disc cover problem \cite{Masuyama1981}. 
	He further presented a technically-involved $O(n^2)$-time algorithm for a special case of two-colored input points where one point is red and all other points are blue; his elegant algorithm transforms the consistent subset problem to the problem of covering points with disks which in turn is transformed to the problem of covering a circle with arcs.  
	It has been recently proved, by Khodamoradi \etal \cite{Khodamoradi2018}, that the consistent subset problem with two colors is also NP-complete---the proof is by a reduction from the planar rectilinear monotone 3-SAT \cite{Berg2012}. Observe that the one color version of the problem is trivial because every single point is a consistent subset. More recently, Banerjee \etal \cite{Banerjee2018} showed that the consistent subset problem on collinear points, i.e., points that lie on a straight line, can be solved optimally in $O(n^2)$ time. 
	
	Recently, Gottlieb~\etal \cite{Gottlieb2018} studied a two-colored version of the consistent subset problem --- referred to as the nearest
	neighbor condensing problem --- where the points come from a metric space. They prove a lower bound for the hardness of approximating a minimum consistent subset; this lower bound includes two parameters: the doubling dimension of the space and the ratio of the minimum distance between points of opposite colors to the diameter of the point set. Moreover, for this two-colored version of the problem, they give an approximation algorithm whose ratio almost matches the lower bound. 
	
	In a related problem, which is called the selective subset problem, the goal is to find the smallest subset $S$ of $P$ such that for every $p\in P_i$ the nearest neighbor of $p$ in $S\cup (P\setminus P_i)$ belongs to $P_i$. Wilfong \cite{Wilfong1992} showed that this problem is also NP-complete even with two colors. See \cite{Banerjee2018} for some recent progress on this problem. 
	
	In this paper we study the consistent subset problem. We improve some previous results and present some new results. To obtain these results, we combine tools from planar separators, additively-weighted Voronoi diagrams with respect to a convex distance function, point location in farthest-point Voronoi diagrams, range trees, paraboloid lifting, minimum covering of a circle with arcs, and several geometric transformations. We present the first subexponential-time algorithm for this problem. We also present an $O(n\log n)$-time algorithm that finds a consistent subset of size two in two-colored point sets (if such a subset exists); this is obtained by transforming the consistent subset problem into a point-cone incidence problem in dimension three. 
	Towards our proof of this running time we present a deterministic $O(n \log n)$-time algorithm for computing a variant of the compact Voronoi diagram; this improves the $O(n \log n)$ expected running time of the randomized algorithm of Bhattacharya \etal \cite{Bhattacharya2010}. We also revisit the case where one point is red and all other points are blue; we give an $O(n\log^2 n)$-time algorithm for this case, thereby improving the previous $O(n^2)$ running time of \cite{Wilfong1992}.
	For collinear points, we present an $O(n)$-time algorithm; this improves the previous running time by a factor of $\Theta(n)$. We also present a non-trivial $O(n^6)$-time dynamic programming algorithm for points arranged on two parallel lines.  
	
	\section{A Subexponential Algorithm}
	
	The consistent subset problem can easily be solved in exponential time by simply checking all possible subsets of $P$. In this section we present the first subexponential-time algorithm for this problem. We consider the decision version of this problem in which we are given a set $P$ of $n$ colored points in the plane and an integer $k$, and we want to decide whether or not $P$ has a consistent subset of size $k$. Moreover, if the answer is positive, then we want to find such a subset. This problem can be solved in time $n^{O(k)}$ by checking all possible subsets of size $k$. We show how to solve this problem in time $n^{O(\sqrt{k})}$; we use a recursive separator-based technique that was introduced in 1993 by Hwang \etal \cite{Hwang1993} for the Euclidean $k$-center problem, and then extended by Marx and Pilipczuk \cite{Marx2015} for planar facility location problems. Although this technique is known before, its application in our setting is not straightforward and requires technical details which we give in this section.

	Consider an optimal solution $S$ of size $k$. The Voronoi diagram of $S$, say $\mathcal{V}$, is a partition of the plane into convex regions. We want to convert $\mathcal{V}$ to a 2-connected 3-regular planar graph that have a balanced curve separator. Then we want to use this separator to split the problem into two subproblems that can be solved independently. To that end, first we introduce small perturbation
	
	 \begin{wrapfigure}{r}{1.8in} 
		\centering
		\vspace{2pt} 
		\includegraphics[width=1.75in]{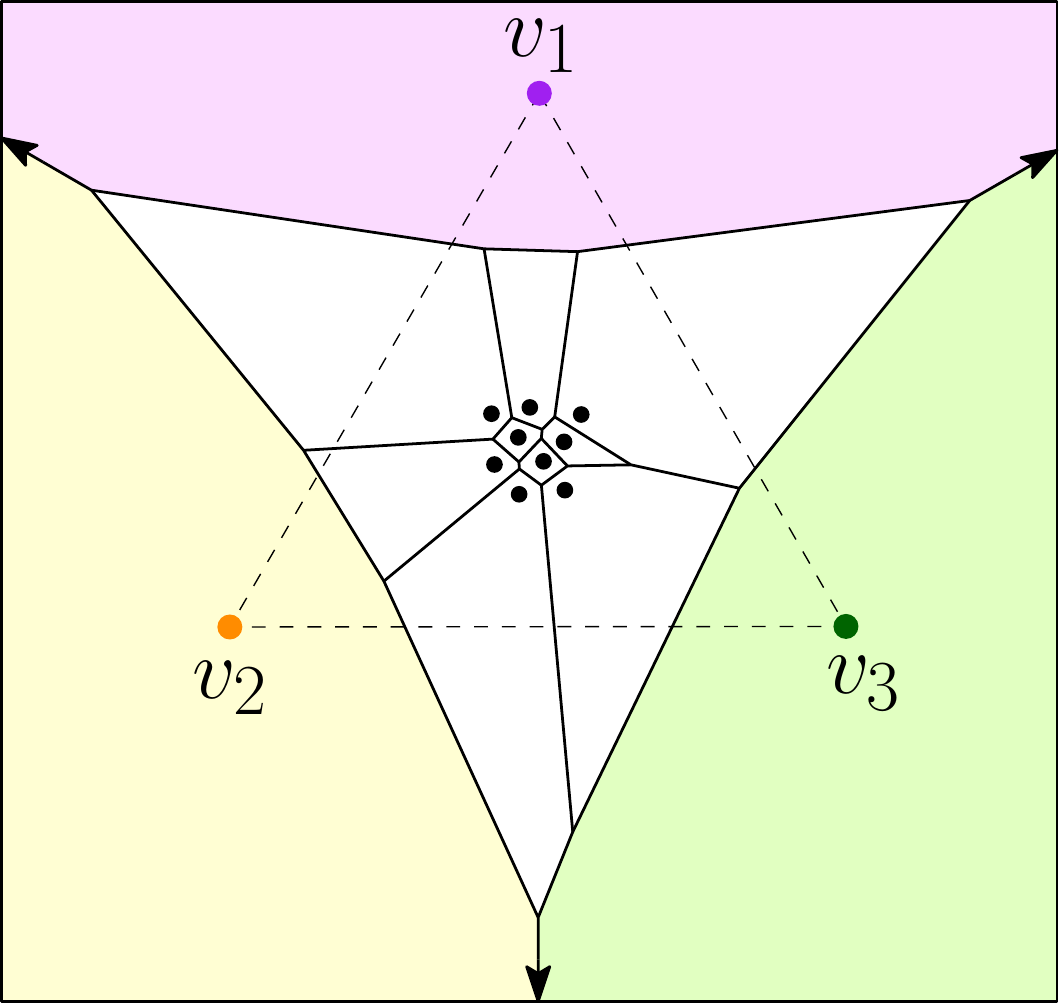} 
		\vspace{-8pt} 
	\end{wrapfigure}
	
	\noindent to the coordinates of points of $P$ to ensure that no four points lie on the boundary of a circle; this ensures that every vertex of $\mathcal{V}$ has degree 3. The Voronoi diagram $\mathcal{V}$ consists of finite segments and infinite rays. We want $\mathcal{V}$ to have at most three infinite rays. To achieve this, we introduce three new points $v_1, v_2,v_3$ that lie on the vertices of a sufficiently large equilateral triangle\footnote{The triangle is large in the sense that for every point $p\in P$, the closet point to $p$, among $P\cup\{v_1,v_2,v_3\}$, is in $P$.} that contains $P$, and then we color them by three new colors; see the right figure. Since these three points have distinct colors, they appear in any consistent subset of $P\cup\{v_1,v_2,v_3\}$. Moreover, since they are far from the original points, by adding them to any consistent subset of $P$ we obtain a valid consistent subset for $P\cup\{v_1,v_2,v_3\}$. Conversely, by removing these three points from any consistent subset of $P\cup\{v_1,v_2,v_3\}$ we obtain a valid consistent subset for $P$.
	Therefore, in the rest of our description we assume, without loss of generality, that $P$ contains $v_1,v_2,v_3$.
	Consequently, the optimal solution $S$ also contains those three points; this implies that $\mathcal{V}$ has three infinite rays which are introduced by $v_1, v_2,v_3$ (see the above figure). We introduce a new vertex at infinity and connect these three rays to that vertex. To this end we obtain a 2-connected 3-regular planar graph, namely $\mathcal{G}$. Marx and Pilipczuk \cite{Marx2015} showed that such a graph has a polygonal separator $\delta$ of size $O(\sqrt{k})$ (going through $O(\sqrt{k})$ faces and vertices) that is {\em face balanced}, in the sense that there are at most $2k/3$ faces of $\mathcal{G}$ strictly inside $\delta$ and at most $2k/3$ faces of $\mathcal{G}$ strictly outside $\delta$. The vertices of $\delta$ alternate between points of $S$ and the vertices of $\mathcal{G}$ as depicted in Figure~\ref{Voronoi-fig}(a). See \cite{Miller1986} for an alternate way of computing a balanced curve separator.

\begin{figure}[htb]
	\centering
	\setlength{\tabcolsep}{0in}
	$\begin{tabular}{cc}
	\multicolumn{1}{m{.5\columnwidth}}{\centering\includegraphics[width=.48\columnwidth]{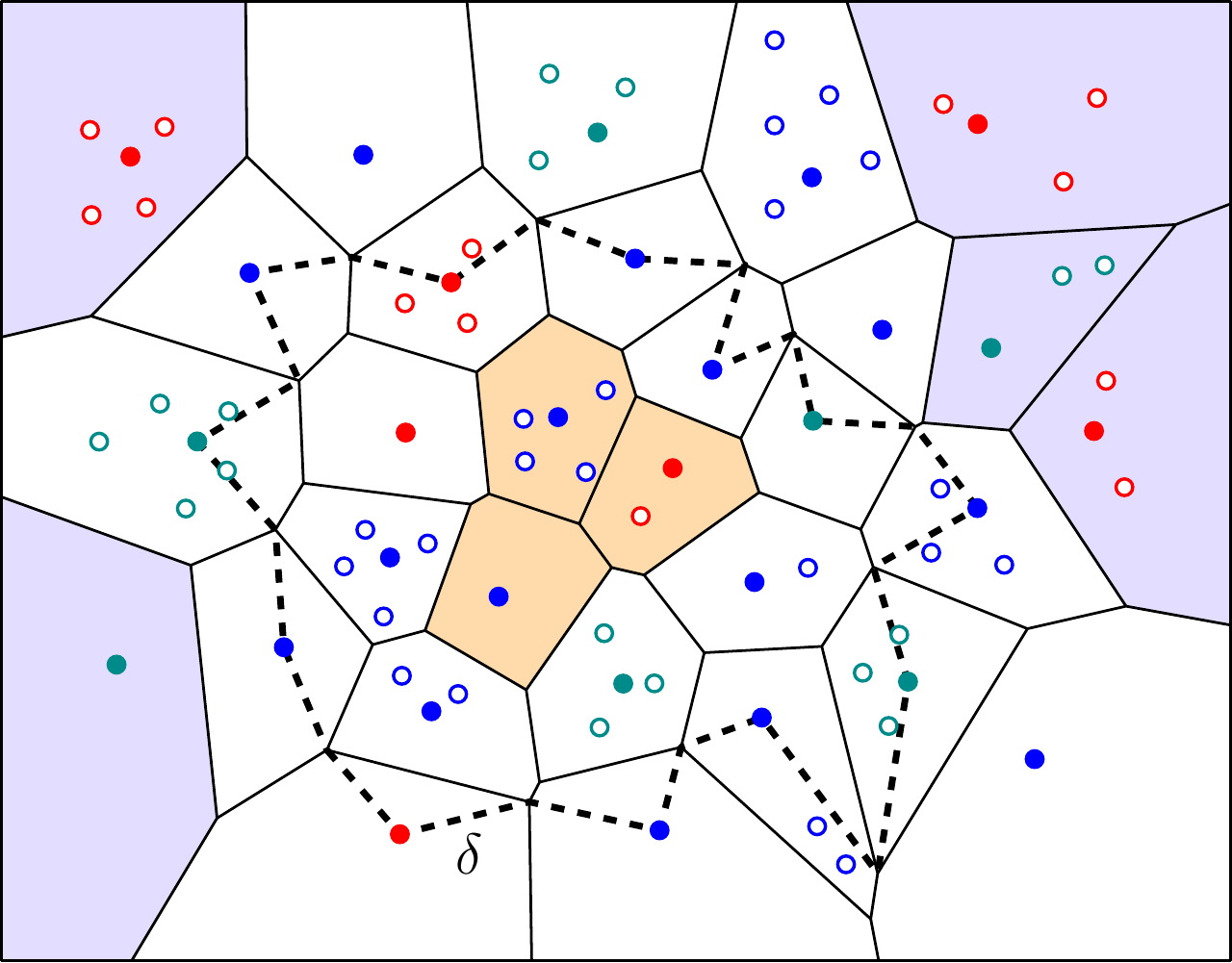}}
	&\multicolumn{1}{m{.5\columnwidth}}{\centering\includegraphics[width=.48\columnwidth]{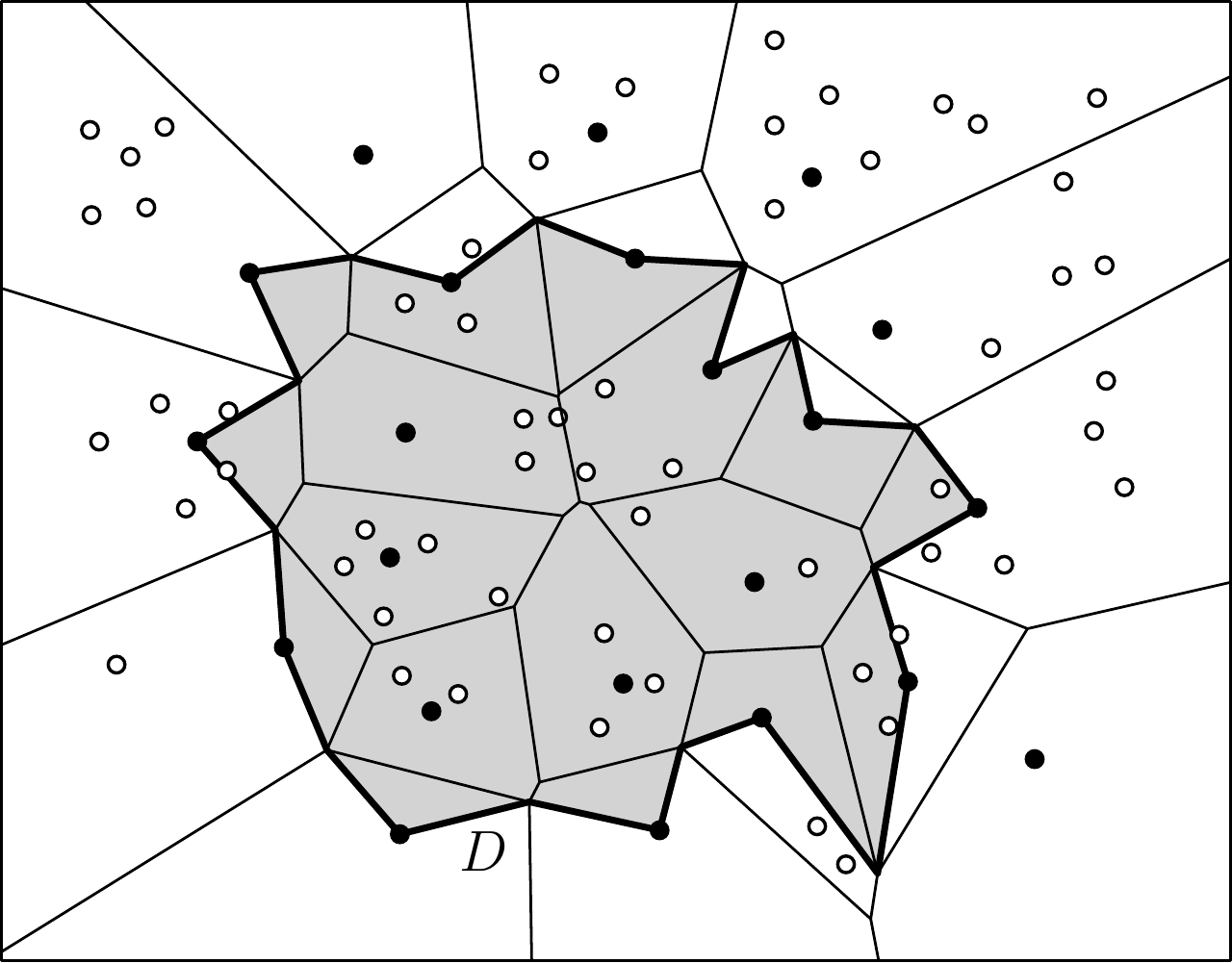}}\\
	(a)&(b)
	\end{tabular}$
	\caption{(a) A solution $S$ (bold points), together with its Voronoi diagram $\mathcal{V}$, and a balanced curve separator $\delta$. (b) A subproblem with input domain $D$ (shaded region) and a set $S'$ (bold points) that is part of the solution.}
	\label{Voronoi-fig}
\end{figure}

	We are going to use dynamic programming based on balanced curve separators of $\mathcal{G}$. The main idea is to use $\delta$ to split the problem into two smaller subproblems, one inside $\delta$ and one outside $\delta$, and then solve each subproblem recursively. But, we do not know $\mathcal{G}$ and hence we have no way of computing $\delta$. However, we can guess $\delta$ by trying all possible balanced curve separators of size $k'=O(\sqrt{k})$.

Every vertex of $\delta$ is either a point of $P$ or a vertex of $\mathcal{G}$ (and consequently a vertex of $\mathcal{V}$) that is introduced by three points of $P$. Therefore, every curve separator of size $k'$ is defined by at most $3k'$ points of $P$, and thus, the number of such separators is at most ${n\choose 3k'}\leqslant n^{3k'}=n^{O(\sqrt{k})}$. To find these curve separators, we try every subset of at most $3k'$ points of $P$. For every such subset we compute its Voronoi diagram, which has at most $6k'$ vertices. For the set that is the union of the $3k'$ points and the $6k'$ vertices, we check all $2^{(6k'+3k')}$ subsets and choose every subset that forms a balanced curve separator (that alternates between points and vertices). Therefore, in a time proportional to $n^{3k'}\cdot 2^{9k'}=n^{O(\sqrt{k})}$ we can compute all balanced curve separators. 
	
	By trying all balanced curve separators, we may assume that we have correctly guessed $\delta$ and the subset $S'$ of $P$, with $|S'|\leqslant 3k'$, that defines $\delta$. The solution of our main problem consists of $S'$ and the solutions of the two separate subproblems, one inside $\delta$ and one outside $\delta$. To solve these two subproblems recursively, in the later steps, we get subproblems of the following form. Throughout our description, we will assume that $P$ is fixed for all subproblems. The input of every subproblem consists of a positive integer $x$ $(\leqslant k)$, a subset $S'$ of $y$ $(\leqslant k)$ points of $P$ that are already chosen to be in the solution, and a polygonal domain $D$---possibly with holes---of size $\Theta(y)$ which is a polygon its vertices alternating between the points of $S'$ and the vertices of the Voronoi diagram of $S'$. The task is to select a subset $S\subseteq (P\cap D)\setminus S'$ of size $x$ such that:
	
	\begin{enumerate}[$(i)$]
		\item $D$ is a polygon where its vertices alternate between the points of $S'$ and the vertices of the Voronoi diagram of $S\cup S'$, and
		\item $S\cup S'$ is a consistent subset for $(P\cap D)\cup S'$.
	\end{enumerate}

	See Figure~\ref{Voronoi-fig}(b) for an illustration of such a subproblem.
	The top-level subproblem has $x=k$ and $y=0$. We stop the recursive calls as soon as we reach a subproblem with $x=O(\sqrt{k})$, in which case, we spend $O(n^x)$ time to solve this subproblem; this is done by trying all subsets of $(P\cap D)\setminus S'$ that have size $x$. For every subproblem, the number of points in $S'$ (i.e., $y$) is at most three times the number of vertices on the boundary of the domain $D$. The number of vertices on the boundary of $D$---that are accumulated during recursive calls---is at most
	 	
	$$\sqrt{k}+\sqrt{\frac{2}{3}k}+\sqrt{\left(\frac{2}{3}\right)^2 k}+\sqrt{\left(\frac{2}{3}\right)^3 k}+...=O(\sqrt{k}).$$
	
	Therefore, $y=|S'|=O(\sqrt{k})$, and thus the Voronoi diagram of $S\cup S'$ has a balanced curve separator of size $O(\sqrt{x+y})=O(\sqrt{k})$.\footnote{In fact the 2-connected 3-regular planar graph obtained from the Voronoi diagram of $S\cup S'$ has such a separator.} We try all possible $n^{O(\sqrt{k})}$ such separators, and for each of which we recursively solve the two subproblems in its interior and exterior. For these two subproblems to be really independent we include the $O(\sqrt{k})$ points, defining the separator, in the inputs of both subproblems. Therefore, the running time of our algorithm can be interpreted by the following recursion
	\[
	T(n,k) \leqslant n^{O(\sqrt{k})} 
	\cdot \max \bigl\{ T(n,k_1+y)+T(n,k_2+y) \mid k_1+k_2+y=k,~ k_1,k_2\leqslant 2k/3,~y=O(\sqrt{k}) \bigr\},
	\]
	which solves to $T(n,k) \leqslant n^{O(\sqrt{k})}$. Notice that our algorithm solves the decision version of the consistent subset problem for a fixed $k$.

	To compute the consistent subset of minimum cardinality,
	whose size, say $k$, is unknown at the start of the algorithm, 
	we apply the following standard technique:
	Start with a constant value $\kappa$, for example $\kappa = 1$.
	Run the decision algorithm with the value $\kappa$. If the
	answer is negative, then double the value of $\kappa$ and repeat this 
	process until the first time the decision algorithm gives a positive 
	answer.
	
	Consider the last value for $\kappa$. Note that $\kappa/2 < k \leqslant \kappa$.
	We perform a binary search for $k$ in the interval $[\kappa/2,\kappa]$.
	In this way, we find the value of $k$, as well as the consistent subset 
	of minimum cardinality, by running the decision algorithm $O(\log \kappa)$ 
	times. Thus, the total running time is 
	$n^{O(\sqrt{\kappa})} \cdot O(\log \kappa)$, which is $n^{O(\sqrt{k})}$. 
	We have proved the following theorem.
	
	\begin{theorem}
		A minimum consistent subset of $n$ colored points in the plane can be computed in $n^{O(\sqrt{k})}$ time, where $k$ is the size of the minimum consistent subset.
	\end{theorem}

\section{Consistent Subset of Size Two}	 
In this section we investigate the existence of a consistent subset of size two in a set of bichromatic points where every point is colored by one of the two colors, say red and blue. Before stating the problem formally we introduce some terminology. For a set $P$ of points in the plane, we denote the convex hull of $P$ by $\CH{P}$. For two points $p$ and $q$ in the plane, we denote the straight-line segment between $p$ and $q$ by $pq$, and the perpendicular bisector of $pq$ by $\bisector{p}{q}$.

Let $R$ and $B$ be two disjoint sets of total $n$ points in the plane such that the points of $R$ are colored red and the points of $B$ are colored blue. We want to decide whether or not $R\cup B$ has a
consistent subset of size two. Moreover, if the answer is positive, then we want to find such points, i.e., a red point $r\in R$ and a blue point $b\in B$ such that all red points are closer to $r$ than to $b$, and all blue points are closer to $b$ than to $r$. Alternatively, we want to find a pair of points $(r,b)\in R\times B$ such that $\bisector{r}{b}$ separates $\CH{R}$ and $\CH{B}$. This problem can be solved in $O(n^2\log n)$ time by trying all the $O(n^2)$ pairs $(r,b)\in R\times B$; for each pair $(r,b)$ we can verify, in $O(\log n)$ time, whether or not $\bisector{r}{b}$ separates $\CH{R}$ and $\CH{B}$. 
In this section we show how to solve this problem in time $O(n\log n)$. To that end, we assume that $\CH{R}$ and $\CH{B}$ are disjoint, because otherwise there is no such pair $(r,b)$.

\begin{wrapfigure}{r}{1.3in} 
	\centering
	\vspace{0pt} 
	\includegraphics[width=1.25in]{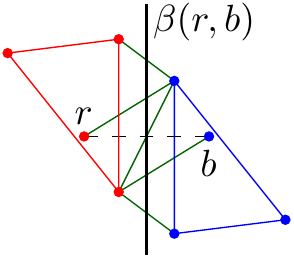} 
	\vspace{-8pt} 
\end{wrapfigure}
It might be tempting to believe that a solution of this problem contains points only from the boundaries of $\CH{R}$ and $\CH{B}$. However, this is not necessarily the case; in the figure to the right, the only solution of this problem contains $r$ and $b$ which are in the interiors of $\CH{R}$ and $\CH{B}$. Also, due to the close relation between Voronoi diagrams and Delaunay triangulations, one may believe that a solution is defined by the two endpoints of an edge in the Delaunay triangulation of $R\cup B$. This is not necessarily the case either; the green edges in the figure to the right, which are the Delaunay edges between $R$ and $B$, do not introduce any solution.

A {\em separating common tangent} of two disjoint convex polygons, $P_1$ and $P_2$, is a line $\ell$ that is tangent to both $P_1$ and $P_2$ such that $P_1$ and $P_2$ lie on different sides of $\ell$. Every two disjoint convex polygons have two separating common tangents; see Figure~\ref{tangents-fig}. Let $\ell_1$ and $\ell_2$ be the separating common tangents of $\CH{R}$ and $\CH{B}$. Let $R'$ and $B'$ be the subsets of $R$ and $B$ on the boundaries of $\CH{R}$ and $\CH{B}$, respectively, that are between $\ell_1$ and $\ell_2$ as depicted in Figure~\ref{tangents-fig}. For two points $p$ and $q$ in the plane, let $D(p,q)$ be the closed disk that is centered at $p$ and has $q$ on
its boundary. 

\begin{lemma}
	\label{inclusion-exclusion-lemma}
	For every two points $r\in R$ and $b\in B$, the bisector $\bisector{r}{b}$ separates $R$ and $B$ if and only if
	\begin{enumerate}[$(i)$]
		\item $\forall r'\in R':~~~ b\notin D(r',r)$, and
		\item $\forall b'\in B':~~~ b\in D(b',r)$.
	\end{enumerate}
\end{lemma}
\begin{proof}
	For the direct implication since $\bisector{r}{b}$ separates $R$ and $B$, every red point $r'$ (and in particular every point in $R'$) is closer to $r$ than to $b$; this implies that $D(r',r)$ does not contain $b$ and thus (i) holds. Also, every blue point $b'$ (and in particular every point in $B'$) is closer to $b$ than to $r$; this implies that $D(b',r)$ contains $b$ and thus (ii) holds. See Figure~\ref{tangents-fig}.
	
	Now we prove the converse implication by contradiction. Assume that  both (i) and (ii) hold for some $r\in R$ and some $b\in B$, but the bisector $\bisector{r}{b}$ does not separate $R$ and $B$. After a suitable rotation we may assume that $\bisector{r}{b}$ is vertical, $r$ is to the left side of $\bisector{r}{b}$ and $b$ is to the right side of $\bisector{r}{b}$. Since $\bisector{r}{b}$ does not separate $R$ and $B$, there exists either a point of $R$ to the right side of $\bisector{r}{b}$, or a point of $B$ to the left side of $\bisector{r}{b}$. If there is a point of $R$ to the right side of $\bisector{r}{b}$ then there is also a point $r'\in R'$ to the right side of $\bisector{r}{b}$. In this case $r'$ is closer to $b$ than to $r$, and thus the disk $D(r',r)$ contains $b$ which contradicts (i). If there is a point of $B$ to the left side of $\bisector{r}{b}$ then there is also a point $b'\in B'$ to the left side of $\bisector{r}{b}$. In this case $b'$ is closer to $r$ than to $b$ and thus the disk $D(b',r)$ does not contain $b$ which contradicts (ii).  
\end{proof}	

\begin{figure}[htb]
	\centering
	\includegraphics[width=.5\columnwidth]{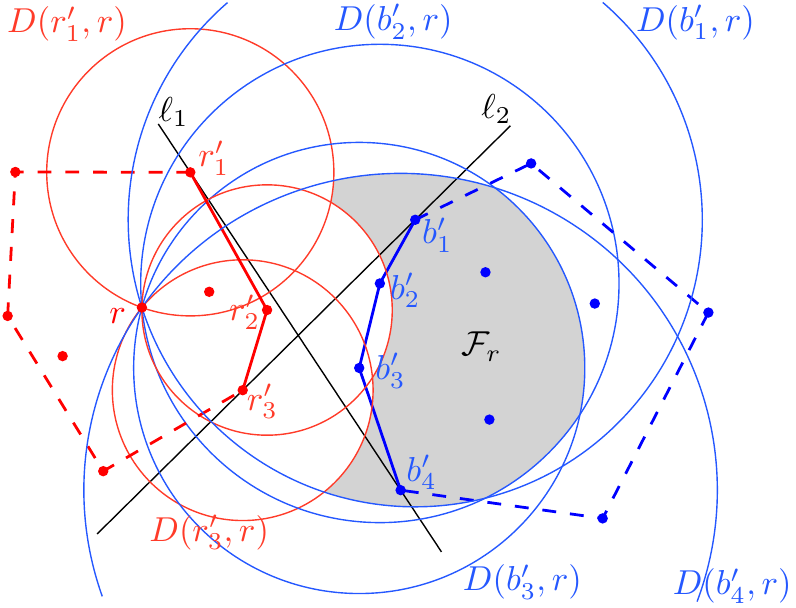}
	\caption{The lines $\ell_1$ and $\ell_2$ are the separating common tangents of $\CH{R}$ and $\CH{B}$. $R'=\{r'_1,r'_2,r'_3\}$ and $B'=\{b'_1,b'_2,b'_3,b'_4\}$ are the subsets of $R$ and $B$ on boundaries of $\CH{R}$ and $\CH{B}$ that lie between $\ell_1$ and $\ell_2$. The feasible region $F_r$ for point $r$ is shaded.}
	\label{tangents-fig}
\end{figure}

Lemma~\ref{inclusion-exclusion-lemma} implies that for a pair $(r,b)\in R\times B$ to be a consistent subset of $R\cup B$ it is necessary and sufficient that every point of $R'$ is closer to $r$ than to $b$, and every point of $B'$ is closer to $b$ than to $r$. This lemma does not imply that $r$ and $b$ are necessarily in $R'$ and $B'$. 	
Observe that Lemma~\ref{inclusion-exclusion-lemma} holds even if we swap the roles of $r, r', R'$ with $b,b',B'$ in (i) and (ii). Also, observe that this lemma holds even if we take $R'$ and $B'$ as all red and blue points on boundaries of $\CH{R}$ and $\CH{B}$. 

For every red point $r\in R$ we define a {\em feasible region} $\mathcal{F}_r$ as follow
\[
\mathcal{F}_r ~=~ \left(\bigcap_{b'\in B'} D(b',r)\right)
\setminus\left( \bigcup_{r'\in R'} D(r',r) \right).
\]
See Figure~\ref{tangents-fig} for illustration of a feasible region. Lemma~\ref{inclusion-exclusion-lemma}, together with this definition, imply the following corollary.
\begin{corollary}
	\label{feasible-cor}
	For every two points $r\in R$ and $b\in B$, the bisector $\bisector{r}{b}$ separates $R$ and $B$ if and only if $b\in \mathcal{F}_r$.
\end{corollary}

Based on this corollary, our original decision problem reduces to the following question.
\begin{question}
	\label{q1}
	Is there a blue point $b\in B$ such that $b$ lies in the feasible region $\mathcal{F}_r$ of some red point $r\in R$?
\end{question}
If the answer to Question~\ref{q1} is positive then $\{r,b\}$ is a consistent subset for $R\cup B$, and if the answer is negative then $R\cup B$ does not have a consistent subset with two points.	
In the rest of this section we show how to answer Question~\ref{q1}. To that end, we lift the plane onto the paraboloid $z=x^2+y^2$ by projecting every point $s=(x,y)$ in $\RR^2$ onto the point $\hat s=(x,y, x^2+y^2)$ in $\RR^3$. This lift projects a circle in $\RR^2$ onto a plane in $\RR^3$. Consider a disk $D(p,q)$ in $\RR^2$ and let $\pi(p,q)$ be the plane in $\RR^3$ that contains the projection of the boundary circle of $D(p,q)$. Let $H^-(p,q)$ be the lower closed halfspace defined by $\pi(p,q)$, and let $H^+(p,q)$ be the upper open halfspace defined by $\pi(p,q)$. For every point $s\in \RR^2$, its projection $\hat s$ lies in $H^-(p,q)$ if and only if $s\in D(p,q)$, and lies in $H^+(p,q)$ otherwise. Moreover, $\hat s$ lies in $\pi(p,q)$ if and only if $s$ is on the boundary circle of $D(p,q)$. For every point $r\in R$ we define a polytope $\mathcal{C}_r$ in $\RR^3$ as follow
\[
\mathcal{C}_r ~=~ \left(\bigcap_{b'\in B'} H^-(b',r)\right)
\cap \left( \bigcap_{r'\in R'} H^+(r',r) \right).
\]

Based on the above discussion, Corollary~\ref{feasible-cor} can be translated to the following corollary.

\begin{corollary}
	\label{polytope-cor}
	For every two points $r\in R$ and $b\in B$, the bisector $\bisector{r}{b}$ separates $R$ and $B$ if and only if $\hat b\in \mathcal{C}_r$.
\end{corollary}

This corollary, in turn, translates Question~\ref{q1} to the following question.
\begin{question}
	\label{q2}
	Is there a blue point $b\in B$ such that its projection $\hat b$ lies in the polytope $\mathcal{C}_r$ for some red point $r\in R$?
\end{question}

Now, we are going to answer Question~\ref{q2}. The polytope $\mathcal{C}_r$ is the intersection of some halfspaces, each of which has $\hat r$ on its boundary plane. Therefore, $\mathcal{C}_r$ is a cone in $\RR^3$ with apex $\hat r$; see Figure~\ref{cones-fig}. Recall that $|R\cup B|=n$, however, for the purposes of worst-case running-time analysis and to simplify indexing, we will index the red points, and also the blue points, from 1 to $n$. Let $r_1,r_2,,\dots, r_n$ be the points of $R$. For every point $r_i\in R$, let $\tau_i$ be the translation that brings $\hat r_1$ to $\hat r_i$. Notice that $\tau_1$ is the identity transformation. In the rest of this section we will write $\mathcal{C}_{i}$ for $\mathcal{C}_{r_i}$.

\begin{lemma}
	\label{translation-lemma}
	For every point $r_i\in R$, the cone $\mathcal{C}_{i}$ is the translation of $\mathcal{C}_{1}$ with respect to $\tau_i$. 
\end{lemma}
\begin{wrapfigure}{r}{2.2in} 
	\vspace{0pt} 
	\centering
	\includegraphics[width=2.1in]{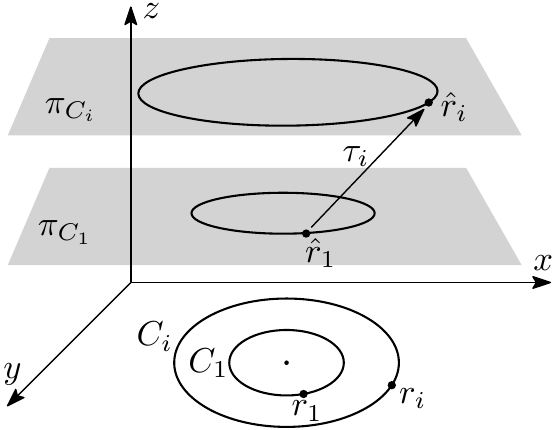} 
	\vspace{-5pt} 
\end{wrapfigure}
\noindent{\em Proof.} 
For a circle $C$ in $\RR^2$, let $\pi_C$ denote the plane in $\RR^3$ that $C$ translates onto.
For every two concentric circles $C_1$ and $C_i$ in $\RR^2$ it holds that $\pi_{C_1}$ and $\pi_{C_i}$ are parallel; see the figure to the right. It follows that, if $C_1$ passes through the point $r_1$, and $C_i$ passes through the point $r_i$, then $\pi_{C_i}$ is obtained from $\pi_{C_1}$ by the translation $\tau_i$ that brings $\hat r_1$ to $\hat r_i$, that is $\tau_i(\pi_{C_1})=\pi_{C_i}$. A similar argument holds also for the halfspaces defined by $\pi_{C_1}$ and $\pi_{C_i}$. Since for every $a\in R'\cup B'$ the disks $D(a,r_1)$ and $D(a,r_i)$ are concentric and the boundary of $D(a,r_1)$ passes through $r_1$ and the boundary of $D(a,r_i)$ passes through $r_i$, it follows that $\tau_i(H^+(a,r_1))=H^+(a,r_i)$ and $\tau_i(H^-(a,r_1))=H^-(a,r_i)$. Since a translation of a polytope is obtained by translating each of the halfspaces defining it, we have $\tau_i(\mathcal{C}_{1})= \mathcal{C}_{i}$ as depicted in Figure~\ref{cones-fig}.
\qed 

\vspace{10pt}

\begin{figure}[H]
	\centering
	\vspace{5pt}
	\includegraphics[width=.7\columnwidth]{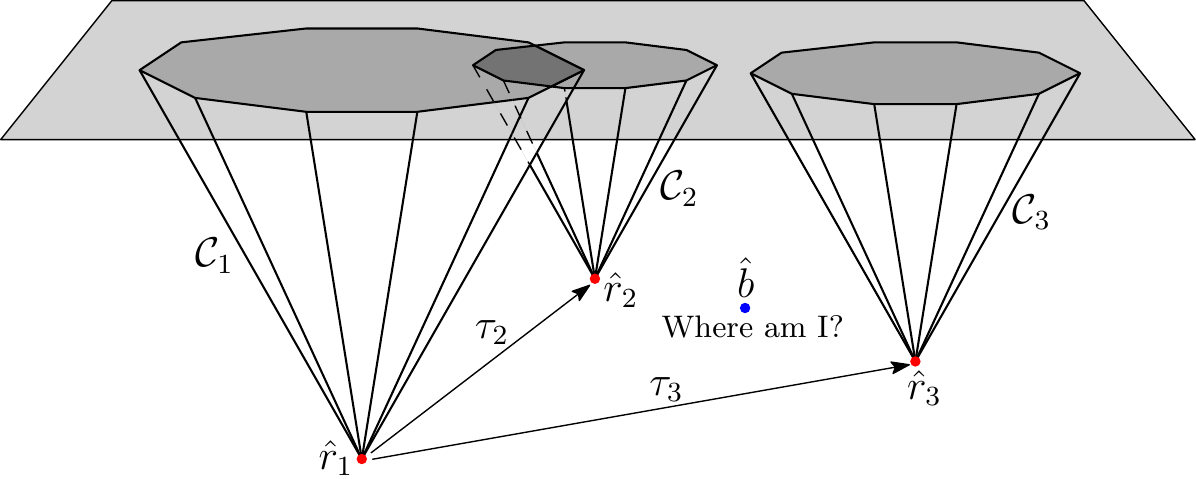}
	\caption{The cones $\mathcal{C}_{2}$ and $\mathcal{C}_{3}$ are the translations of $\mathcal{C}_{1}$ with respect to $\tau_2$ and $\tau_3$.}
	\label{cones-fig}
\end{figure}

It follows from Lemma~\ref{translation-lemma} that to answer Question~\ref{q2} it suffices to solve the following problem: Given a cone $\mathcal{C}_{1}$ defined by $n$ halfspaces, $n$ translations of $\mathcal{C}_{1}$, and set of $n$ points, we want to decide whether or not there is a point in some cone (see Figure~\ref{cones-fig}). This can be verified in $O(n\log n)$ time, using Theorem~\ref{point-cone-thr} that we will prove later in Section~\ref{point-cone-section}. This is the end of our constructive proof. The following theorem summarizes our result in this section.

\begin{theorem}
	Given a set of $n$ bichromatic points in the plane, in $O(n \log n)$ time, we can compute a consistent subset of size two $($if such a set exists$)$. 
\end{theorem} 

\section{One Red Point} 
In this section we revisit the consistent subset problem for the case where one input point is red and all other points are blue. Let $P$ be a set of $n$ points in the plane consisting of a red point and $n-1$ blue points.
Observe that any consistent subset of $P$ contains the only red point and some blue points. In his seminal work in SoCG 1991, Wilfong \cite{Wilfong1992} showed that $P$ has a consistent subset of size at most seven (including the red point); this implies an $O(n^6)$-time brute force algorithm for this problem.
Wilfong showed how to solve this problem in $O(n^2)$-time; his elegant algorithm transforms the consistent subset problem to the problem of covering points with disks which in turn is transformed to the problem of covering a circle with arcs. The running time of his algorithm is dominated by the transformation to the circle covering problem which involves computation of $n-1$ arcs in $O(n^2)$ time; all other transformations together with the solution of the circle covering problem take $O(n\log n)$ time (\cite[Lemma 19 and Theorem 9]{Wilfong1992}).

We first introduce the circle covering problem, then we give a summary of Wilfong's transformation to this problem, and then we show how to perform this transformation in $O(n \log^2 n)$ time which  implies the same running time for the entire algorithm. We emphasis that the most involved part of the algorithm, which is the correctness proof of this transformation, is due to Wilfong.

\begin{figure}[htb]
	\centering
	\setlength{\tabcolsep}{0in}
	$\begin{tabular}{cc}
	\multicolumn{1}{m{.5\columnwidth}}{\centering\includegraphics[width=.42\columnwidth]{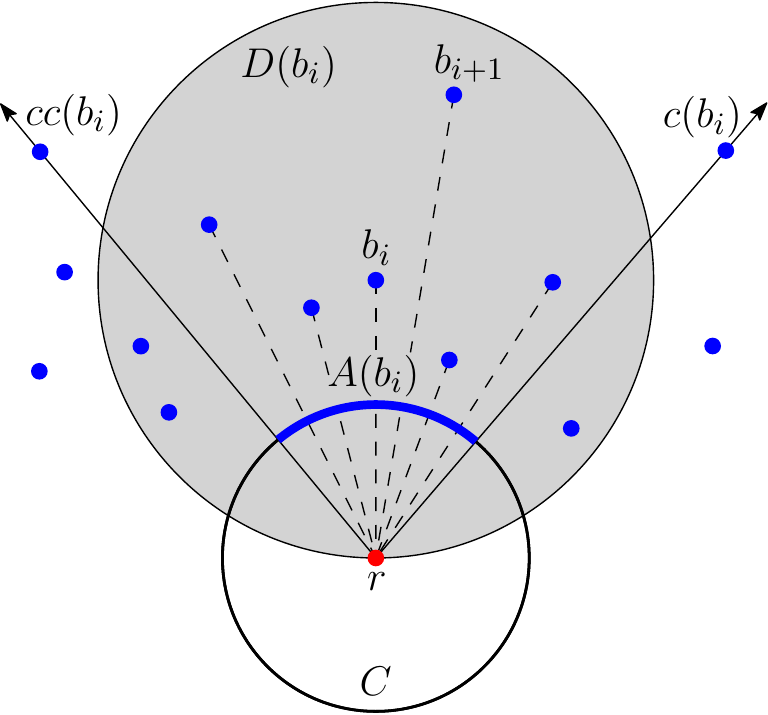}}
	&\multicolumn{1}{m{.5\columnwidth}}{\centering\includegraphics[width=.42\columnwidth]{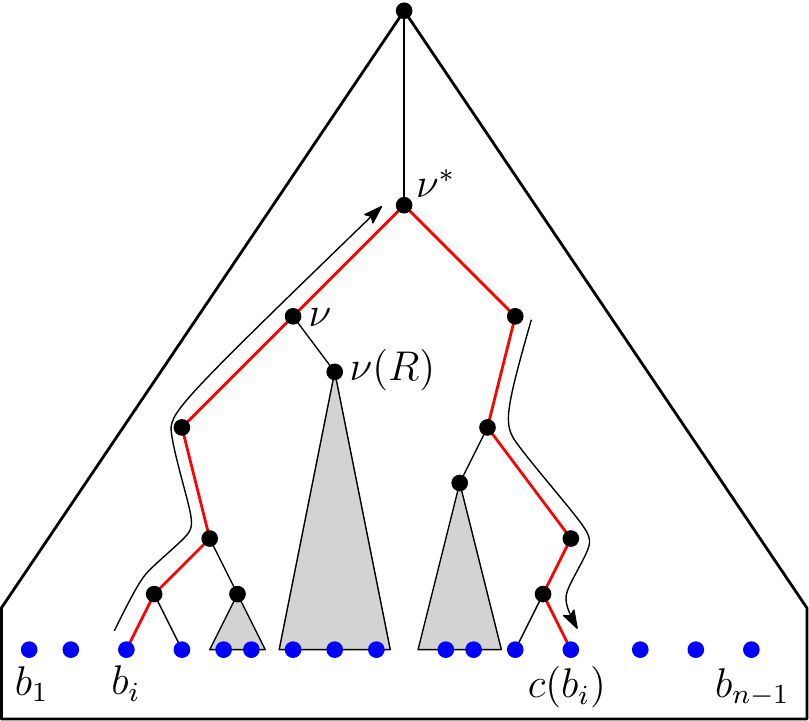}}\\
	(a)&(b)
	\end{tabular}$
	\caption{(a) Transformation to the circle covering problem. (b) The range tree $T$ on blue points.}
	\label{one-red-fig}
\end{figure}

Let $C$ be a circle and let $\mathcal{A}$ be a set of arcs covering the entire $C$. The {\em circle covering} problem asks for a subset of $\mathcal{A}$, with minimum cardinality, that covers the entire $C$.

Wilfong's algorithm starts by mapping input points to the projective plane, and then transforming (in two stages) the consistent subset problem to the circle covering problem. Let $P$ denote the set of points after the mapping, and let $r$ denote the only red point of $P$. The transformation, which is depicted in Figure~\ref{one-red-fig}(a), proceeds as follows. Let $C$ be a circle centered at $r$ that does not contain any blue point. Let $b_1,b_2,\dots,b_{n-1}$ be the blue points in clockwise circular order around $r$ ($b_1$ is the first clockwise point after $b_{n-1}$, and $b_{n-1}$ is the first counterclockwise point after $b_1$). For each point $b_i$, let $D(b_i)$ be the disk of radius $|rb_i|$ centered at $b_i$. 
Define $cc(b_i)$ to be the first counterclockwise point (measured from $b_i$) that is not in $D(b_i)$, and similarly define $c(b_i)$ to be the first clockwise point that is not in $D(b_i)$. Denote by $A(b_i)$ the open arc of $C$ that is contained in the wedge with counterclockwise boundary ray from $r$ to $cc(b_i)$ and the clockwise boundary ray from $r$ to $c(b_i)$.\footnote{Wilfong shrinks the endpoint of $A(b_i)$ that corresponds to $cc(b_i)$ by half the clockwise angle from $cc(b_i)$ to the next point, and shrinks the endpoint of $A(b_i)$ that corresponds to $c(b_i)$ by half the counterclockwise angle from $c(b_i)$ to the previous point.} Let $\mathcal{A}$ be the set of all arcs $A(b_i)$; since blue points are assumed to be in circular order, $\mathcal{A}$ covers the entire $C$. Wilfong proved that our instance of the consistent subset problem is equivalent to the problem of covering $C$ with $\mathcal{A}$. The running time of his algorithm is dominated by the computation of $\mathcal{A}$ in $O(n^2)$ time. We show how to compute $\mathcal{A}$ in $O(n\log^2 n)$ time.

In order to find each arc $A(b_i)$ it suffices to find the points $cc(b_i)$ and $c(b_i)$. Having the clockwise ordering of points around $r$, one can find these points in $O(n)$ time for each $b_i$, and consequently in $O(n^2)$ time for all $b_i$'s. In the rest of this section we show how to find $c(b_i)$ for all $b_i$'s in $O(n\log^2 n)$ time; the points $cc(b_i)$ can be found in a similar fashion.

By the definition of $c(b_i)$ all points of the sequence $b_{i+1}, \dots,c(b_i)$, except $c(b_i)$, lie inside $D(b_i)$. Therefore among all points $b_{i+1},\dots, c(b_i)$, the point $c(b_i)$ is the farthest from $b_i$. This implies that in the farthest-point Voronoi diagram of $b_{i+1},\dots,c(b_i)$, the point $b_i$ lies in the cell of $c(b_i)$. To exploit this property of $c(b_i)$, we construct a 1-dimensional range tree $T$ on all blue points based on their clockwise order around $r$; blue points are stored at the leaves of $T$ as in Figure~\ref{one-red-fig}(b). At every internal node $\nu$ of $T$ we store the farthest-point Voronoi diagram of the blue points that are stored at the leaves of the subtree rooted at $\nu$; we refer to this diagram by FVD($\nu$). This data structure can be computed in $O(n\log^2 n)$ time because $T$ has $O(\log n)$ levels and in each level we compute farthest-point Voronoi diagrams of total $n-1$ points. To simplify our following description, at the moment we assume that $b_1,\dots, b_{n-1}$ is a linear order. At the end of this section, in Remark 1, we show how to deal with the circular order. 

We use the above data structure to find each point $c(b_i)$ in $O(\log^2 n)$ time. To that end, we walk up the tree from the leaf containing $b_i$ (first phase), and then walk down the tree (second phase) as described below; also see Figure~\ref{one-red-fig}(b). For every internal node $\nu$, let $\nu(L)$ and $\nu(R)$ denote its left and right children, respectively. In the first phase, for every internal node $\nu$ in the walk, we locate the point $b_i$ in FVD($\nu(R)$) and find the point $b_f$ that is farthest from $b_i$. If $b_f$ lies in $D(b_i)$ then also does every point stored at the subtree of $\nu(R)$. In this case we continue walking up the tree and repeat the above point location process until we find, for the first time, the node $\nu^*$ for which $b_f$ does not lie in $D(b_i)$. To this end we know that $c(b_i)$ is among the points stored at $\nu^*(R)$. Now we start the second phase and walk down the tree from $\nu^*(R)$. For every internal node $\nu$ in this walk, we locate $b_i$ in FVD($\nu(L)$) and find the point $b_f$ that is farthest from $b_i$. If $b_f$ lies in $D(b_i)$, then also does every point stored at $\nu(L)$, and hence we go to $\nu(R)$, otherwise we go to $\nu(L)$. At the end of this phase we land up in a leaf of $T$, which stores $c(b_i)$. The entire walk has $O(\log n)$ nodes and at every node we spend $O(\log n)$ time for locating $b_i$. Thus the time to find $c(b_i)$ is $O(\log^2 n)$. Therefore, we can find all $c(b_i)$'s in $O(n\log^2 n)$ total time.   

\begin{theorem}
	A minimum consistent subset of $n$ points in the plane, where one point is red and all other points are blue, can be computed in $O(n \log^2 n)$ time. 
\end{theorem} 

{\noindent \bf Remark 1.} To deal with the circular order $b_1,\dots, b_{n-1}$, we build the range tree $T$ with $2(n-1)$ leaves $b_1,\dots,b_{n-1},b_1,\dots, b_{n-1}$.  
For a given $b_{i}$, the point $c(b_i)$ can be any of the points $b_{i+1},\dots, b_{n-1},\allowbreak  b_{1},\dots,b_{i-1}$. To find $c(b_i)$, we first follow the path from the root of $T$ to the leftmost leaf that stores $b_i$, and then from that leaf we start looking for $c(b_i)$ as described above.

\section{Restricted Point Sets}
In this section we present polynomial-time algorithms for the consistent subset problem on three restricted classes of point sets. First we present an $O(n)$-time algorithm for collinear points; this improves the previous quadratic-time algorithm of Banerjee \etal \cite{Banerjee2018}. Then we present an involved non-trivial $O(n^6)$-time dynamic programming algorithm for points that are placed on two parallel lines. Finally we present an $O(n^4)$-time algorithm for two-colored points, namely red and blue, that are placed on two parallel lines such that all points on one line are red and all points on the other line are blue.

\subsection{Collinear Points}
\label{collinear-section}
Let $P$ be a set of $n$ colored points on the $x$-axis, and let $p_1,\dots,p_n$ be the sequence of these points from left to right. We present a dynamic programming algorithm that solves the consistent subset problem on $P$. To simplify the description of our algorithm we add a point $p_{n+1}$ very far (at distance at least $|p_1p_n|$) to the right of $p_n$. We set the color of $p_{n+1}$ to be different from that of $p_{n}$. Observe that every solution for $P\cup \{p_{n+1}\}$ contains $p_{n+1}$. Moreover, by removing $p_{n+1}$ from any optimal solution of $P\cup \{p_{n+1}\}$ we obtain an optimal solution for $P$. 
Therefore, to compute an optimal solution for $P$, we first compute an optimal solution for $P\cup \{p_{n+1}\}$ and then remove $p_{n+1}$.

Our algorithm maintains a table $T$ with $n+1$ entries $T(1),\dots, T(n+1)$. Each table entry $T(k)$ represents the number of points in a minimum consistent subset of $P_k=\{p_1,\dots,p_k\}$ provided that $p_k$ is in this subset. The number of points in an optimal solution for $P$ will be $T(n+1)-1$; the optimal solution itself can be recovered from $T$. In the rest of this section we show how to solve a subproblem with input $P_{k}$ provided that $p_k$ should be in the solution (thereby in the rest of this section the phrase ``solution of $P_k$'' refers to a solution that contains $p_k$). In fact, we show how to compute $T(k)$, by a bottom-up dynamic programming algorithm that scans the points from left to right. If $P_k$ is monochromatic, then the optimal solution contains only $p_k$, and thus, we set $T(k)=1$. Hereafter assume that $P_k$ is not monochromatic. Consider the partition of $P_k$ into maximal blocks of consecutive points such that the points in each block have the same color. Let $B_1,B_2,\dots,B_{m-1},B_m$ denote these blocks from left to right, and notice that $p_k$ is in $B_m$. Assume that the points in $B_m$
are red and the points in $B_{m-1}$ are blue.
Let $p_{y}$ be the leftmost point in $B_{m-1}$; see Figure~\ref{collinear-fig}(a). Any optimal solution for $P_k$ contains at least one point from $\{p_y,\dots, p_{k-1}\}$; let $p_i$ be the rightmost such point ($p_i$ can be either red or blue). Then, $T(k)=T(i)+1$. Since we do not know the index $i$, we try all possible values in $\{y,\dots, k-1\}$ and select one that produces a {\em valid} solution, and that minimizes $T(k)$:
\begin{equation}
\notag
T(k)=\min\{T(i)+1 \mid {i\in\{y,\dots,k-1\} \text{ and $i$ produces a valid solution}}\}.
\end{equation} 
The index $i$ produces a valid solution (or $p_i$ is valid) if one of the following conditions hold:

\begin{enumerate}[$(i)$]
	\item $p_i$ is red, or
	\item $p_i$ is blue, and for every $j\in \{i+1,\dots, k-1\}$ it holds that if $p_j$ is blue then $p_j$ is closer to $p_i$ than to $p_k$, and if $p_j$ is red then $p_j$ is closer to $p_k$ than to $p_i$.
\end{enumerate}
	
If $(i)$ holds then $p_i$ and $p_k$ have the same color. In this case the validity of our solution for $P_k$ is ensured by the validity of the solution of $P_i$. If $(ii)$ holds then $p_i$ and $p_k$ have distinct colors. In this case the validity of our solution for $P_k$ depends on the colors of points $p_{i+1},\dots,p_{k-1}$. To verify the validity in this case, it suffices to check the colors of only two points that are to the left and to the right of the mid-point of the segment $p_ip_k$. This can be done in $O(|B_{m-1}|)$ time for all blue points in $B_{m-1}$ while scanning them from left to right. Thus, $T(k)$ can be computed in $O(k)$ time because $|B_{m-1}|=O(k)$. Therefore, the total running time of the above algorithm is $O(n^2)$.

\begin{figure}[htb]
	\centering
	\setlength{\tabcolsep}{0in}
	$\begin{tabular}{cc}
	\multicolumn{1}{m{.5\columnwidth}}{\centering\includegraphics[width=.49\columnwidth]{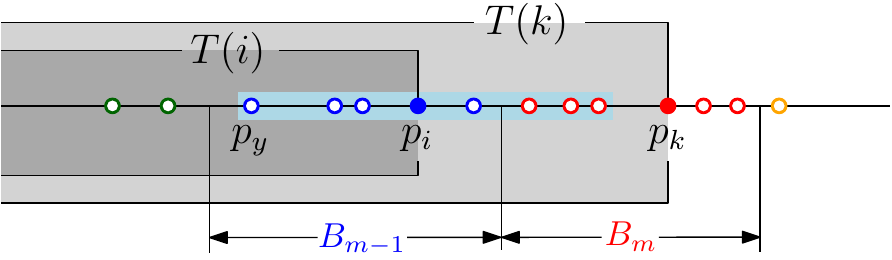}}
	&\multicolumn{1}{m{.5\columnwidth}}{\centering\includegraphics[width=.44\columnwidth]{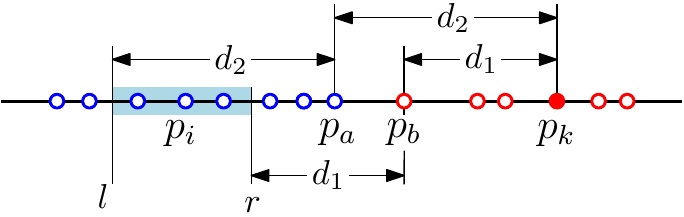}}\\
	(a)&(b)
	\end{tabular}$
	\caption{(a) Illustration of the computation of $T(k)$ from $T(i)$. (b) Any blue point in the range $[l,r]$ is valid.}
	\label{collinear-fig}
\end{figure}	
		
We are now going to show how to compute $T(k)$ in constant time, which in turn improves the total running time to $O(n)$. To that end we first prove the following lemma.

\begin{lemma}
	\label{diff-lemma}
	Let $s\in \{1,\dots,m\}$ be an integer, $p_i,p_{i+1},\dots,p_j$ be a sequence of points in $B_s$, and $x\in\{i,\dots,j\}$ be an index for which $T(x)$ is minimum. Then, $T(j)\leqslant T(x)+1$.
\end{lemma}
\begin{proof}
	To verify this inequality, observe that by adding $p_j$ to the optimal solution of $P_x$ we obtain a valid solution (of size $T(x)+1$) for $P_j$. Therefore, any optimal solution of $P_j$ has at most $T(x)+1$ points, and thus $T(j)\leqslant T(x)+1$.
\end{proof}

At every point $p_j$, in every block $B_s$, we store the index $i$ of the first point $p_i$ to the left of $p_j$ where $p_i\in B_s$ and $T(i)$ is strictly smaller than $T(j)$; if there is no such point $p_i$ then we store $j$ at $p_j$. These indices can be maintained in linear time while scanning the points from left to right. We use these indices to compute $T(k)$ in constant time as described below. 
   
Notice that if the minimum, in the above calculation of $T(k)$, is obtained by a red point in $B_m$ then it always produces a valid solution, but if the minimum is obtained by a blue point then we need to verify its validity. In the former case, it follows from Lemma~\ref{diff-lemma} that the smallest $T(\cdot)$ for red points in $B_m\setminus\{p_k\}$ is obtained either by $p_{k-1}$ or by the point whose index is stored at $p_{k-1}$. Therefore we can find the smallest $T(\cdot)$ in constant time. Now consider the latter case where the minimum is obtained by a blue point in $B_{m-1}$. Let $p_a$ be the rightmost point of $B_{m-1}$, and let $p_b$ be the leftmost endpoint of $B_m$. Set $d_1=|p_bp_k|$ and $d_2=|p_ap_k|$ as depicted in Figure~\ref{collinear-fig}(b). Set $l=x(p_a)-d_2$ and $r=x(p_b)-d_1$, where $x(p_a)$ and $x(p_b)$ are the $x$-coordinates of $p_a$ and $p_b$. Any point $p_i\in B_{m-1}$ that is to the right of $r$ is invalid because otherwise $p_b$ would be closer to $p_i$ than to $p_k$. Any point $p_i\in B_{m-1}$ that is to the left of $l$ is also invalid because otherwise $p_a$ would be closer to $p_k$ than to $p_i$. However, every point $p_i\in B_{m-1}$, that is in the range $[l,r]$, is valid because it satisfies condition $(ii)$ above. Thus, to compute $T(k)$ it suffices to find a point of $B_{m-1}$ in range $[l,r]$ with the smallest $T(\cdot)$. By slightly abusing notation, let $p_r$ be the rightmost point of $B_{m-1}$ in range $[l,r]$. It follows from Lemma~\ref{diff-lemma} that the smallest $T(\cdot)$ is obtained either by $p_r$ or by the point whose index is stored at $p_{r}$. Thus, in this case also, we can find the smallest $T(\cdot)$ in constant time.

It only remains to identify, in constant time, the index that we should store at $p_k$ (to be used in next iterations). If $p_k$ is the leftmost point in $B_m$, then we store $k$ at $p_k$. Assume that $p_k$ is not the leftmost point in $B_m$, and let $x$ be the index stored at $p_{k-1}$. In this case, if $T(x)$ is smaller than $T(k)$ then we store $x$ at $p_k$, otherwise we store $k$. This assignment ensures that $p_k$ stores a correct index. 

Based on the above discussion we can compute $T(k)$ and identify the index at $p_k$ in constant time. 
Therefore, our algorithm computes all values of $T(\cdot)$ in $O(n)$ total time. The following theorem summarizes our result in this section.

\begin{theorem}
	A minimum consistent subset of $n$ collinear colored points can be computed in $O(n)$ time, provided that the points are given from left to right.  
\end{theorem}

\subsection{Points on Two Parallel Lines} 
\label{mix-section}
	 In this section we study the consistent subset problem on points that are placed on two parallel lines.
	 Let $P$ and $Q$ be two disjoint sets of colored points of total size $n$, such that the points of $P$ are on a straight line $L_P$ and points of $Q$ are on a straight line $L_Q$ that is parallel to $L_P$. The goal is to find a minimum consistent subset for $P\cup Q$. We present a top-down dynamic programming algorithm that solves this problem in $O(n^6)$ time. By a suitable rotation and reflection we may assume that $L_P$ and $L_Q$ are horizontal and $L_P$ lies above $L_Q$. If any of the sets $P$ and $Q$ is empty, then this problem reduces to the collinear version that is discussed in Section~\ref{collinear-section}. Assume that none of $P$ and $Q$ is empty. An optimal solution may contain points from only $P$, only $Q$, or from both $P$ and $Q$. We consider these three cases and pick one that gives the minimum number of points:
		\begin{enumerate}
			\item {\em The optimal solution contains points from only $Q$.} Consider any solution $S\subseteq Q$. For every point $p\in P$, let $p'$ be the vertical projection of $p$ on $L_Q$. Then, a point $s\in S$ is the closest point to $p$ if and only if $s$ is the closest point to $p'$. This observation suggests the following algorithm for this case: First project all points of $P$ vertically on $L_Q$; let $P'$ be the resulting set of points. Then, solve the consistent subset problem for points in $Q\cup P'$, which are collinear on $L_Q$, with this invariant that the points of $P'$ should not be included in the solution but should be included in the validity check. This can be done in $O(n)$ time by modifying the algorithm of Section~\ref{collinear-section}. 
			
			\item {\em The optimal solution contains points from only $P$.} The solution of this case is analogous to that of previous case.
			
			\item {\em The optimal solution contains points from both $P$ and $Q$.} The description of this case is more involved. Add two dummy points $p^-$ and $p^+$ at $-\infty$ and $+\infty$ on $L_P$, respectively. Analogously, add $q^-$ and $q^+$ on $L_Q$. Color these four points by four new colors that are different from the colors of points in $P\cup Q$. See Figure~\ref{mix-high-level-fig}. Set $D=\{p^+,p^-,q^+,q^-\}$. Observe that every solution for $P\cup Q\cup D$ contains all points of $D$. Moreover, by removing $D$ from any optimal solution of $P\cup Q\cup D$ we obtain an optimal solution for $P\cup Q$. Therefore, to compute an optimal solution for $P\cup Q$, we first compute an optimal solution for $P\cup Q\cup D$ and then remove $D$. In the rest of this section we show how to compute an optimal solution for $P\cup Q\cup D$. Without loss of generality, from now on, we assume that $p^-$ and $p^+$ belong to $P$, and $q^-$ and $q^+$ belong to $Q$. For a point $p$ let $\ell_p$ be the vertical line through $p$. 
			
			\begin{figure}[htb]
				\centering
				\includegraphics[width=.7\columnwidth]{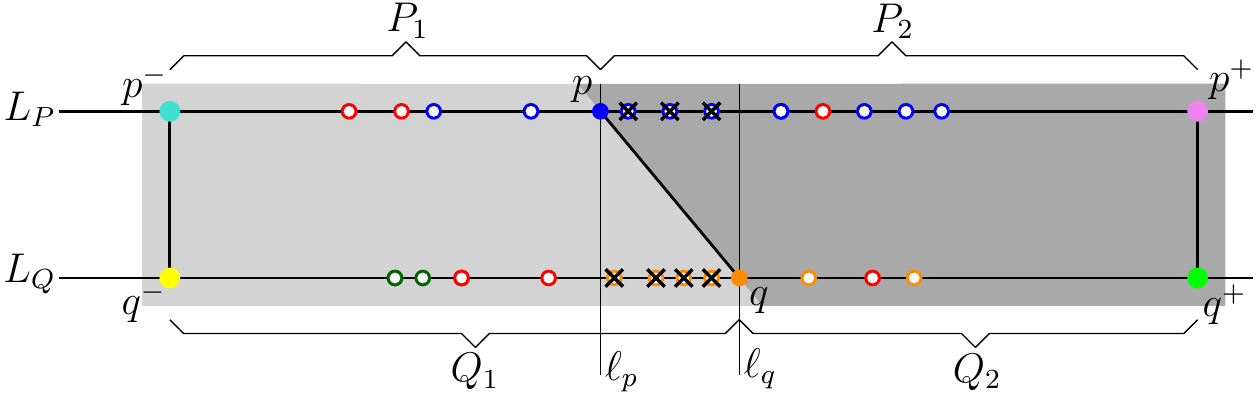}
				\caption{The pair $(p,q)$ is the closest pair in the optimal solution where $p\in P\setminus\{p^+,p^-\}$ and $q\in Q\setminus\{q^+,q^-\}$. This pair splits the problem into two independent subproblems.}
				\label{mix-high-level-fig}
			\end{figure}
			
			In the following description the term ``solution" refers to an optimal solution. Consider a solution for this problem with input pair $(P,Q)$, and let $p$ and $q$ be the closest pair in this solution such that $p\in P\setminus\{p^+,p^-\}$ and $q\in Q\setminus\{q^+,q^-\}$ (for now assume that such a pair exists; later we deal with all different cases). These two points split the problem into two subproblems $(P_1,Q_1)$ and $(P_2,Q_2)$ where $P_1$ contains all points of $P$ that are to the left of $p$ (including $p$), $P_2$ contains all points of $P$ that are to the right of $p$ (including $p$), and $Q_1, Q_2$ are defined analogously. Our choice of $p$ and $q$ ensures that no point in the solution lies between the vertical lines $\ell_p$ and $\ell_q$ because otherwise that point would be part of the closest pair. See Figure~\ref{mix-high-level-fig}. Thus, $(P_1,Q_1)$ and $(P_2,Q_2)$ are independent instances of the problem in the sense that for any point in $P_1\cup Q_1$ (resp. $P_2\cup Q_2$) its closest point in the solution belongs to $P_1\cup Q_1$ (resp. $P_2\cup Q_2$). Therefore, if $p$ and $q$ are given to us, we can solve $(P,Q)$ as follows: First we recursively compute a solution for $(P_1,Q_1)$ that contains $p^-,q^-,p,q$ and does not contain any point between $\ell_p$ and $\ell_q$. We compute an analogous solution for $(P_2,Q_2)$ recursively. Then, we take the union of these two solutions as our solution of $(P,Q)$. We do not know $p$ and $q$, and thus we try all possible choices. 
			
			Let $p_1,p_2,\dots,p_{|P|}$ and $q_1,q_2,\dots,q_{|Q|}$ be the points of $P$ and $Q$, respectively, from left to right, where $p_1=p^-$ and $q_1=q^-$. In later steps in our recursive solution we get subproblems of type $S(i,j,k,l)$ where the input to this subproblem is $\{p_i,\dots,p_j\}\cup\{q_k,\dots,q_l\}$ and we want to compute a minimum consistent subset that
			\begin{itemize}
				\item contains $p_i, p_j,q_k$, and $q_l$, and
				\item does not contain any point between $\ell_{p_i}$ and $\ell_{q_k}$, nor any point between $\ell_{p_j}$ and $\ell_{q_l}$. 
			\end{itemize}
			To simplify our following description, we may also refer to $S(\cdot)$ as a four dimensional matrix where each of its entries stores the size of the solution for the corresponding subproblem; the solution itself can also be retrieved from $S(\cdot)$. The solution of the original problem will be stored in $S(1,|P|,1,|Q|)$. In the rest of this section we show how to solve $S(i,j,k,l)$ by a top-down dynamic programming approach. Let $p_{i'}$ and $q_{k'}$ be the first points of $P$ and $Q$, respectively, that are to the right sides of both $\ell_{p_i}$ and $\ell_{q_k}$, and let $p_{j'}$ and $q_{l'}$ be the first points of $P$ and $Q$, respectively, that are to the left sides of both $\ell_{p_j}$ and $\ell_{q_l}$; see Figure~\ref{mix-fig-1}. Depending on whether or not the solution of $S(i,j,k,l)$ contains points from $\{p_{i'},\dots,p_{j'}\}$ and $\{q_{k'},\dots,q_{l'}\}$ we consider the following three cases and pick one that minimizes $S(i,j,k,l)$.
			\begin{enumerate}
				\item {\em The solution does not contain points from any of $\{p_{i'},\dots,p_{j'}\}$ and $\{q_{k'},\dots,q_{l'}\}$.} Thus, the solution contains only $p_i$, $p_j$, $q_k$, and $q_l$. To handle this case, we verify the validity of $\{p_i, p_j,q_k, q_l\}$. If this set is a valid solution, then we assign $S(i,j,k,l)=4$, otherwise we assign $S(i,j,k,l)=+\infty$.
				\item {\em The solution contains points from both $\{p_{i'},\dots,p_{j'}\}$ and $\{q_{k'},\dots,q_{l'}\}$.} Let $p_s\in\{p_{i'},\allowbreak \dots,\allowbreak p_{j'}\}$ and $q_t\in\{q_{k'},\dots,q_{l'}\}$ be two such points with minimum distance. Our choice of $p_s$ and $q_t$ ensures that no point of the solution lies between $\ell_{p_s}$ and $\ell_{q_t}$; see Figure~\ref{mix-fig-1}. Therefore, the solution of $S(i,j,k,l)$ is the union of the solutions of subproblems $S(i,s,k,t)$ and $S(s,j,t,l)$. Since we do not know $s$ and $t$, we try all possible pairs and pick one that minimizes $S(\cdot)$, that is 
				$$S(i,j,k,l)=\min\{S(i,s,k,t)+S(s,j,t,l)-2\mid i'\leqslant s \leqslant j',~k' \leqslant t \leqslant l'\},$$
				
				where ``$-2$" comes from the fact that $p_s$ and $q_t$ are counted twice. The validity of this solution for $S(i,j,k,l)$ is ensured by the validity of the solutions of $S(i,s,k,t)$ and $S(s,j,t,l)$, and the fact that these solutions do not contain any point between $\ell_{p_s}$ and $\ell_{q_t}$.
				
				\begin{figure}[htb]
					\centering
					\includegraphics[width=.5\columnwidth]{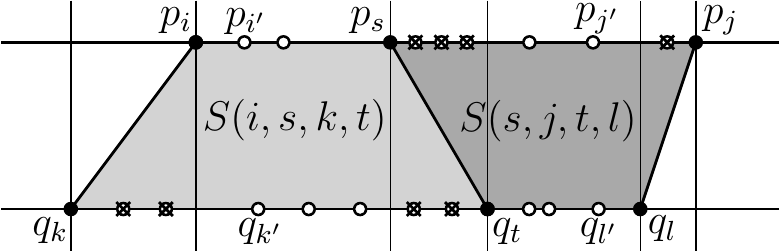}
					\caption{$(p_s,q_t)$ is the closest pair in the solution where $s\in\{i',\dots,j'\}$ and $t\in\{k',\dots,l'\}$.}
					\label{mix-fig-1}
				\end{figure}
				
				\item {\em The solution contains points from $\{q_{k'},\dots,q_{l'}\}$ but not from $\{p_{i'},\dots,p_{j'}\}$, or vice versa.} Because of symmetry, we only describe how to handle the first case. If the solution contains exactly one point form $\{q_{k'},\dots,q_{l'}\}$, then we can easily solve this subproblem by trying every point $q_t$ in this set and pick one for which $\{p_i,p_j,q_k,q_t,q_l\}$ is valid solution, then we set $S(i,j,k,l)=5$. Hereafter assume that the solution contains at least two points from $\{q_{k'},\dots,q_{l'}\}$. Let $q_s$ and $q_t$ be the leftmost and rightmost such points, respectively. Consider the Voronoi diagram of $p_i,q_k,q_s$ and the Voronoi diagram of $p_j,q_l,q_t$. Depending on whether or not the Voronoi cells of $q_k$ and $q_l$ intersect the line segment $p_ip_j$ we consider the following two cases. 
				
				\begin{enumerate}
					\item {\em The Voronoi cell of $q_k$ or the Voronoi cell of $q_l$ does not intersect $p_ip_j$.} Because of symmetry we only describe how to handle the case where the Voronoi cell of $q_k$ does not intersect $p_ip_j$. See Figure~\ref{mix-fig-2}. In this case, $q_k$ cannot be the closest point to any of the points $p_{i+1},\dots, p_{j-1}$, and thus, the solution of $S(i,j,k,l)$ consists of $q_k$ together with the solution of $S(i,j,s,l)$. Since we do not know $s$, we try all possible choices. An index $s\in\{k',\dots,l'-1\}$ is {\em valid} if the Voronoi cell of $q_k$---in the Voronoi diagram of $p_i,q_k,q_s$---does not intersect the line segment $p_ip_j$, and every point in $\{q_{k+1},\dots,q_{s-1}\}$ has the same color as its closest point among $p_i$, $q_k$, and $q_s$. We try all possible choices of $s$ and pick one that is valid and minimizes $S(i,j,k,l)$. Thus,					
						$$S(i,j,k,l)=\min \{S(i,j,s,l)+1\mid i'\leqslant s\leqslant l'-1 \text{ and $s$ is valid}\}.$$  
					
					\begin{figure}[htb]
						\centering
						\setlength{\tabcolsep}{0in}
						$\begin{tabular}{cc}
						\multicolumn{1}{m{.5\columnwidth}}{\centering\includegraphics[width=.35\columnwidth]{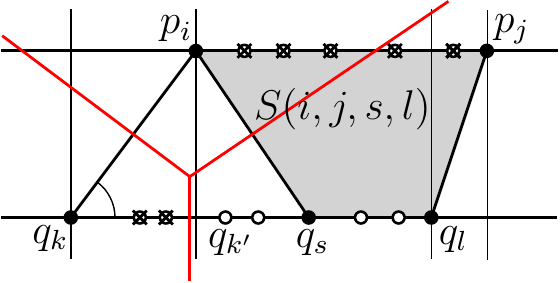}}
						&\multicolumn{1}{m{.5\columnwidth}}{\centering\includegraphics[width=.35\columnwidth]{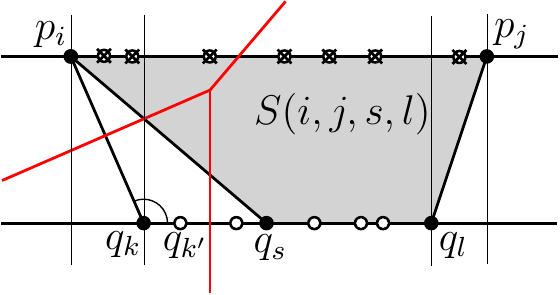}}
						\end{tabular}$
						\caption{The Voronoi cell of $q_k$ does not intersect the line segment $p_ip_j$.}
						\label{mix-fig-2}
					\end{figure}
						
					\item {\em The Voronoi cells of both $q_k$ and $q_l$ intersect $p_ip_j$.} In this case the Voronoi cells of both $q_t$ and $q_s$ also intersect $p_ip_j$; see Figure~\ref{mix-fig-3}.
					In the following description we slightly abuse the notation and refer to the input points $\{p_i,\dots,p_j\}$ and $\{q_k,\dots,q_l\}$ by $P$ and $Q$, respectively. Let $p_{s'}$ be the first point of $P$ to the right of $\ell_{q_s}$, and let $p_{t'}$ be the be the first point of $P$ to the left of $\ell_{q_t}$. Let $P'=\{p_{s'},\dots,p_{t'}\}$ and $Q'=\{q_{s},\dots,q_{t}\}$. Consider any (not necessarily optimal) solution of $S(i,j,k,l)$ that consists of $V=\{p_i, p_j, q_k, q_t, q_s, q_l\}$ and some other points in $\{q_{s+1},\dots,q_{t-1}\}$. The closest point in this solution, to any point of $(P\cup Q)\setminus (P'\cup Q')$, is in $V$. Thus, the (optimal) solution of $S(i,j,k,l)$ consists of $V$ and the optimal solution $S'$ of the consistent subset problem on $P'\cup Q'$ provided that $q_s$ and $q_t$ are in $S'$ and no point of $P'$ is in $S'$. Let $T(s,t)$ denote this new problem on $P'\cup Q'$. We solve $T(s,t)$ by a similar method as in case 1: First we project points of $P'$ on $L_Q$ and then we solve the problem for collinear points. Let $P''$ be the set of projected points. To solve $T(s,t)$, we solve the consistent subset problem for $Q'\cup P''$, which are collinear, with this invariant that the solution contains $q_s$ and $q_t$, and does not contain any point of $P''$; see Figure~\ref{mix-fig-3}. This can be done simply by modifying the algorithm of Section~\ref{collinear-section}. 				
					Therefore, $S(i,j,k,l)=T(s,t)+4$. A pair $(s,t)$ of indices is valid if for every point $x$ in $(P\cup Q)\setminus (P'\cup Q')$ it holds that the color of $x$ is the same as the color of $x$'s closest point in $V$.					
					Since we do not know $s$ and $t$ we try all possible pairs and pick one that is valid and minimizes $S(i,j,k,l)$. Therefore,
					$$S(i,j,k,l)=\min \{T(s,t)+4\mid k< s<t<l \text{ and $(s,t)$ is valid}\}.$$   
					\begin{figure}[htb]
						\centering
						\includegraphics[width=.55\columnwidth]{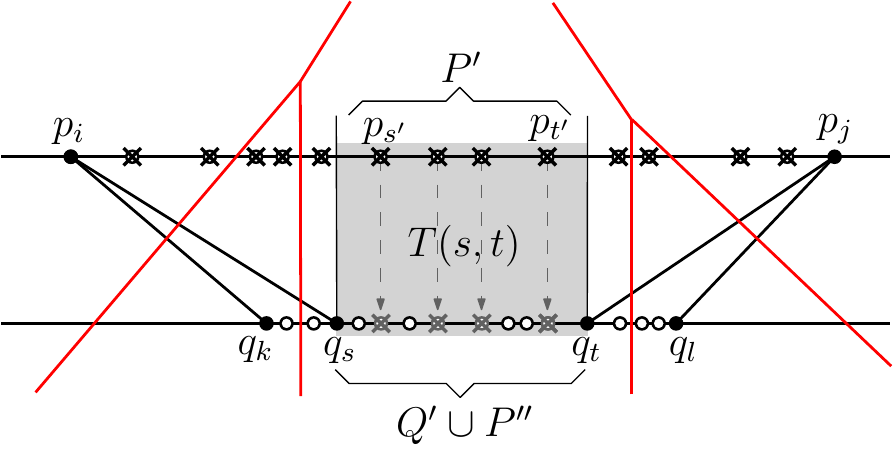}
						\caption{The Voronoi cells of both $q_k$ and $q_l$ intersect the line segment $p_ip_j$.}
						\label{mix-fig-3}
					\end{figure}
				\end{enumerate} 
			\end{enumerate}  
		\end{enumerate}

\paragraph{Running Time Analysis:} Cases 1 and 2 can be handled in $O(n\log n)$ time. Case 3 involves four subcases (a), (b), (c)-i, and (c)-ii. We classify the subproblems in these subcases by types 3(a), 3(b), 3(c)-i, and 3(c)-ii, respectively. The number of subproblems of each type is $O(n^4)$. For every subproblem of type 3(a) we only need to verify the validity of $\{p_i, p_j, q_k, q_l\}$; this can be done in $O(n)$ time. Every subproblem of type 3(b) can be solved in $O(n^2)$ time by trying all pairs $(s,t)$. Every subproblem of type 3(c)-i can be solved in $O(n^2)$ time by trying $O(n)$ possible choices for $s$ and verifying the validity of each of them in $O(n)$ time. 

\begin{wrapfigure}{r}{2.5in} 
	\centering
	\vspace{0pt} 
	\includegraphics[width=2.4in]{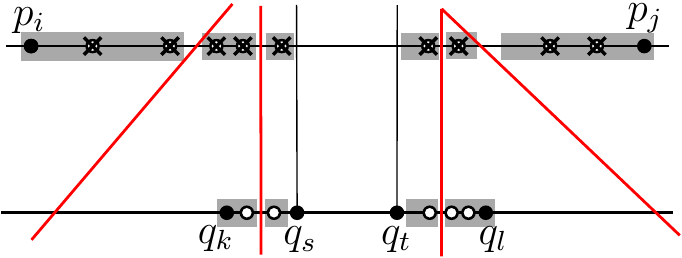} 
	\vspace{-8pt} 
\end{wrapfigure}
Now we show that every subproblem of type 3(c)-ii can also be solved in $O(n^2)$ time. To solve every such subproblem we try $O(n^2)$ pairs $(s,t)$ and we need to verify the validity of every pair. To verify the validity of $(s,t)$ we need to make sure that every point in $(P\cup Q)\setminus (P'\cup Q')$ has the same color as its closest point in $V=\{p_i, p_j, q_k, q_s,q_t,q_l\}$. The Voronoi diagrams of $p_i,q_k, q_s$ and $p_j,q_l, q_t$ together with the lines $\ell_{q_s}$ and $\ell_{q_t}$ partition the points of $(P\cup Q)\setminus (P'\cup Q')$ into 10 intervals, 6 intervals on $L_P$ and 4 intervals on $L_Q$; see the figure to the right. For $(s,t)$ to be feasible it is necessary and sufficient that all points in every interval $I$ have the same color as the point in $V$ that has $I$ in its Voronoi cell. If we know the color of points in each of these 10 intervals, then we can verify the validity of $(s,t)$ in constant time. The total number of such intervals is $O(n^2)$ and we can compute in $O(n^2)$ preprocessing time the color of all of them. Therefore, after $O(n^2)$ preprocessing time we can solve all subproblems of type 3(c)-ii in $O(n^6)$ time. Notice that the total number of subproblems of type $T(s,t)$ in case 3(c)-ii is $O(n^2)$ and we can solve all of them in $O(n^3\log n)$ time before solving subproblems $S(i,j,k,l)$. The following theorem summarizes our result in this section.

\begin{theorem}
	A minimum consistent subset of $n$ colored points on two parallel lines can be computed in $O(n^6)$ time.  
\end{theorem}
	
\subsection{Bichromatic Points on Two Parallel Lines} 
\label{bichromatic-section}
Let $P$ be a set of $n$ points on two parallel lines in the plane such that all points on one line are colored red and all points on the other line are colored blue. We present a top-down dynamic programming algorithm that solves the consistent problem on $P$ in $O(n^4)$ time. By a suitable rotation and reflection we may assume that the lines are horizontal, and the red points lie on the top line. Let $R$ and $B$ denote the set of red and blue points respectively. Let $r_1,\dots,r_{|R|}$ and $b_1,\dots,b_{|B|}$ be the sequences of red points and blue points from left to right, respectively. For each $i\in\{1,\dots |R|\}$ let $R_i$ denote the set $\{r_1,\dots,r_i\}$, and for each $j\in\{1,\dots, |B|\}$ let $B_j$ denote the set $\{b_1,\dots,b_j\}$. For a point $p$ let $\ell_p$ be the vertical line through $p$.

Any optimal solution for this problem contains at least one blue point and one red point. Moreover, the two rightmost points in any optimal solution have distinct colors, because otherwise we could remove the rightmost one and reduce the size of the optimal solution. We solve this problem by guessing the two rightmost points in an optimal solution; in fact we try all pairs $(r_i, b_j)$ where $i\in\{1,\dots |R|\}$ and $j\in\{1,\dots, |B|\}$. For every pair $(r_i, b_j)$ we solve the consistent subset problem on $R_i\cup B_j$ provided that $r_i$ and $r_j$ are in the solution, and no point between the vertical lines $\ell_{r_i}$ and $\ell_{b_j}$ is in the solution (because $r_i$ and $b_j$ are the two rightmost points in the solution). Then, among all pairs $(r_i, b_j)$ we choose one whose corresponding solution is a valid consistent subset for $R\cup B$ and has minimum number of points. The solution corresponding to $(r_i, b_j)$ is a valid consistent subset for $R\cup B$ if for every $x\in\{i+1, \dots,|R|\}$, the point $r_x$ is closer to $r_i$ than to $b_j$, and for every $y\in\{j+1, \dots,|B|\}$, the point $b_y$ is closer to $b_j$ than to $r_i$. To analyze the running time, notice that we guess $O(n^2)$ pairs $(r_i, b_j)$. In the rest of this section we show how to solve the subproblem associated with each pair $(r_i, b_j)$ in $O(n^2)$ time. The validity of the solution corresponding to $(r_i,b_j)$ can be verified in $O(|R|+|B|-i-j)$ time. Therefore, the total running time of our algorithm is $O(n^4)$. 

\begin{figure}[htb]
	\centering
	\setlength{\tabcolsep}{0in}
	$\begin{tabular}{cc}
	\multicolumn{1}{m{.5\columnwidth}}{\centering\includegraphics[width=.37\columnwidth]{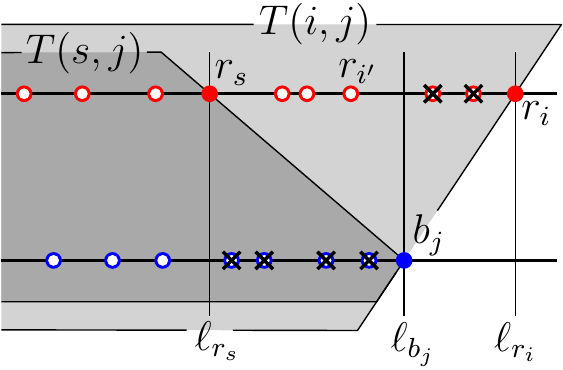}}
	&\multicolumn{1}{m{.5\columnwidth}}{\centering\includegraphics[width=.37\columnwidth]{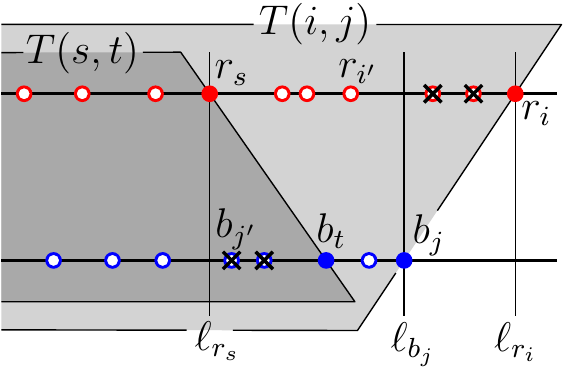}}\\
	(a)&(b)
	\end{tabular}$
	\caption{Illustration of the recursive computation of $T(i,j)$, where (a) $b_j$ is the only blue point in the solution that is to the right of $\ell_{r_s}$, and (b) $b_t$ and $b_j$ are the only two blue points in the solution that are to the right of $\ell_{r_s}$. The crossed points cannot be in the  solution.}
	\label{two-line-fig}
\end{figure}

To solve subproblems associated with pairs $(r_i,b_j)$, we maintain a table $T$ with $|R|\cdot|B|$ entries $T(i,j)$ where $i\in\{1,\dots |R|\}$ and $j\in\{1,\dots, |B|\}$. Each entry $T(i,j)$ represents the number of points in a minimum consistent subset of $R_i\cup B_j$ provided that $r_i$ and $b_j$ are in this subset and no point of $R_i\cup B_j$, that lies between $\ell_{r_i}$ and $\ell_{b_j}$, is in this subset. We use dynamic programming and show how to compute $T(i,j)$ in a recursive fashion. By symmetry we may assume that $r_i$ is to the right of $\ell_{b_j}$. In the following description the term ``solution" refers to an optimal solution associated with $T(i,j)$. Let $r_{i'}$ be the first red point to the left of $\ell_{b_j}$. Observe that if the solution does not contain any red point other than $r_i$, then $\{r_i,b_j\}$ is the solution, i.e., the solution does not contain any other blue point (other than $b_j$) either. Assume that the solution contains some other red points, and let $r_s$, with $s\in\{1,\dots, i'\}$, be the rightmost such point. Let $b_{j'}$ be the first blue point to the right of $\ell_{r_s}$. Now we consider two cases depending on whether or not the solution contains any blue point (other than $b_j$) to the right of $\ell_{r_s}$.

\begin{itemize}
	\item The solution does not contain any other blue point to the right of $\ell_{r_s}$. In this case $T(i,j)=T(s,j) +1$; see Figure~\ref{two-line-fig}(a).
	\item The solution contains some other blue points to the right of $\ell_{r_s}$. Let $b_t$, with $t\in\{j',\dots, j-1\}$, be the rightmost such point. In this case the solution does not contain any blue point that is to the left of $b_t$ and to the right of $\ell_{r_s}$ because otherwise we could remove $b_t$ from the solution. Therefore $T(i,j)=T(s,t)+2$; see Figure~\ref{two-line-fig}(b).
\end{itemize}
Since we do not know $s$ and $t$, we try all possible values and choose one that is {\em valid} and that minimizes $T(i,j)$. Therefore
\[ T(i,j) = \min
\begin{cases}
T(s,j)+1: s\in\{1,\dots, i'\} \text{ and } s \text{ is valid}\\
T(s,t)+2: s\in\{1,\dots, i'\}, t\in\{j',\dots, j-1\} \text{ and } (s,t) \text{ is valid}.
\end{cases}
\]
In the first case, an index $s$ is valid if for every $x\in\{s+1, \dots,i-1\}$ the point $r_x$ is closer to $r_s$ or $r_i$ than to $b_j$. In the second case, a pair $(s,t)$ is valid if for every $x\in\{s+1, \dots,i-1\}$ the point $r_x$ is closer to $r_s$ or $r_i$ than to $b_t$ and $b_j$, and for every $y\in\{t+1, \dots,j-1\}$ the point $b_y$ is closer to $b_t$ or $b_j$ than to $r_s$ and $r_i$. 

To compute $T(i,j)$, we perform $O(n^2)$ look-ups into table $T$, and thus, the time to compute $T(i,j)$ is $O(n^2)$. There is a final issue that we need to address, which is checking the validity of $s$ and $t$ within the same time bound. In the first case we have $O(n)$ look-ups for finding $s$. We can verify the validity of each choice of $s$, in $O(n)$ time, by simply checking the distances of all points in $R'=\{r_{s+1},\dots,r_{i'}\}$ from $r_s$, $r_i$ and $b_j$. Now we consider the second case and describe how to verify, for a fixed $t$, the validity of all pairs $(s,t)$ in $O(n)$ time. First of all observe that in this case, any point $b_y$ with $y\in\{t+1, \dots,j-1\}$, is closer to $b_t$ or $b_j$ than to $r_s$ and $r_i$. Therefore, to check the validity of $(s,t)$ it suffices to consider the points in $R'$. Let $r_{t_1}$ be the first point of $R'$ that is to the left of $\ell_{b_t}$, and let $r_{t_2}$ be the first point of $R'$ that is to the right of $\ell_{b_t}$. Define $r_{j_1}$ and $r_{j_2}$ accordingly but with respect to $\ell_{b_j}$. If there is a point in $R'$ that is closer to $b_t$ than to $r_s$ and $r_i$, then $r_{t_1}$ or $r_{t_2}$ is closer to $b_t$ than to $r_s$ and $r_i$. A similar claim holds for $r_{j_1}$, $r_{j_2}$, and $b_j$. Therefore, to check the validity of $(s,t)$ it suffices to check the distances of $r_{t_1}$, $r_{t_2}$, $r_{j_1}$ and $r_{j_2}$ from the points $r_s$, $r_i$, $b_t$ and $b_j$. This can be done in $O(n)$ time for all $s$ and a fixed $t$. (If any of the points $r_{t_1}$, $r_{t_2}$, $r_{j_1}$ and $r_{j_2}$ is undefined then we do not need to check that point.) The following theorem wraps up this section.

\begin{theorem}
	Let $P$ be a set of $n$ bichromatic points on two parallel lines, such that all points on the same line have the same color. Then, a minimum consistent subset of $P$ can be computed in $O(n^4)$ time.
\end{theorem}  
	
\section{Point-Cone Incidence}
\label{point-cone-section}
In this section we will prove the following theorem.
\begin{theorem}
	\label{point-cone-thr}
	Let $\mathcal{C}$ be a cone in $\RR^3$ with non-empty interior that is given as the intersection of $n$ halfspaces. Given $n$ translations of $\mathcal{C}$ and a set of $n$ points in $\RR^3$, we can decide in $O(n\log n)$ time whether or not there is a point-cone incidence.
\end{theorem}
	
We first provide an overview of the approach and its key ingredients. Let $\mathcal{C}_1,\dots ,\mathcal{C}_n$ be the $n$ cones that are translations of $\mathcal{C}$, and let $P$ denote the set of $n$ input points
that we want to check their incidence with these cones. Consider a direction $d$ such that $\mathcal{C}$ contains an infinite ray from its apex in direction $d$.
After a transformation, we may assume that $d$ is vertically upward.
Consider the lower envelope of the $n$ cones; we want to decide whether there is a point of $P$ above this 
lower envelope; see Figure~\ref{fig:domain}(a). To that end, first we find for every point $p\in P$, the cone $\mathcal{C}_i$ for which the vertical line through $p$ intersects the lower envelope at $\mathcal{C}_i$. Then, we check whether or not $p$ lies above $\mathcal{C}_i$.  

Since all the cones are translations of a common cone,
their lower envelope can be interpreted as a Voronoi diagram with respect to 
a distance function defined by a convex polygon obtaining by intersecting $\mathcal{C}$ with a horizontal plane. Furthermore, in this interpretation, the sites have additive weights that correspond to the vertical shifts in the translations of the cones. Therefore, the lower envelope of the cones can be interpreted as an additively-weighted Voronoi diagram with respect to a convex distance function; the {\em sites} of such diagram are the projections of the apices of the cones into the plane. In order to find the cone (on the lower envelope), that is intersected by the vertical line through $p$, it suffices to locate $p$ in such a Voronoi diagram, i.e., to find $p$'s closest site. 

We adapt the sweep-line approach used by McAllister, Kirkpatrick and Snoeyink~\cite{McAllister1996}
for computing compact Voronoi diagrams for disjoint convex regions with respect to a convex metric.
In a compact Voronoi diagram, one has a linear-size partition of the plane into cells, where each cell
has two possible candidates to be the closest site. Such a structure is enough
to find the closest site to every point $p$: first we locate $p$ in this partition to identify
the cell that contains $p$, and then we compute the distance of $p$ to the two
candidate sites of that cell to find the one that is closer to $p$. The complexity of such a compact Voronoi diagram, in the worst case, is smaller than the complexity of the traditional Voronoi diagram.
Now, we describe our adaption, which involves some modifications of the approach of McAllister \etal \cite{McAllister1996}. Here are the key differences encountered in our adaptation:
\begin{itemize}
	\item The additive weights on the sites can be interpreted as regions defined by 
	convex polygons, but then they are not necessarily disjoint (as required in \cite{McAllister1996}). 
	\item In our case, the Voronoi vertices can be computed faster because the metric and the site regions (encoding the weights) are defined by the same polygon. 
	\item By splitting $\mathcal{C}$ into two cones that have direction $d$ on their boundaries, 
	we can assume that the sweep line and the front line (also referred to as the beach line) coincide; this makes the computation of the Voronoi diagram easier. 
	\item Since the query points (the points of $P$) are already known, we do not need to make a data structure
	for point location or to construct the compact Voronoi diagram
	explicitly. It suffices to make point location on the 
	front line (which is the sweep line in our case) when it passes over a point of $P$.
\end{itemize}

Notice that some of the cones can be contained in some other, and thus, do not appear on the lower envelope of the cones. Bhattacharya et al. \cite{Bhattacharya2010} claimed a randomized algorithm to find, in $O(n\log n)$ expected time, the apices of the cones that appear in the lower envelope of the cones. They discussed a randomized incremental construction, which
is also an adaptation of another algorithm also presented by 
McAllister, Kirkpatrick and Snoeyink~\cite{McAllister1996}. 
Nevertheless, a number of aspects in the construction of \cite{Bhattacharya2010} are not clear.
Our approach is deterministic, and also solves their problem in $O(n \log n)$ worst-case time.

Now we provide the details of our adapted approach.
Consider the cone $\mathcal{C}$ and let $a$ be its apex. 
Let $r$ be a ray emanating from $a$ in the interior of $\mathcal{C}$ such that the plane $\pi$, 
that is orthogonal to $r$ at $a$, intersects $\mathcal{C}$ only in $a$.
We make a rigid motion where the apex of $\mathcal{C}$ becomes the origin,
$r$ becomes vertical, and $\pi$ becomes the horizontal plane defined by $z=0$.
The ray $r$, the plane $\pi$, and the geometric transformation can be computed in 
$O(n)$ time using linear programming in fixed dimension~\cite{Megiddo1984}.
From now on, we will assume that the input is actually given after the transformation.
Let $\mathcal{C}'$ be the intersection of $\mathcal{C}$ with the halfspace $x\geqslant 0$,
and let $\mathcal{C}''$ be the intersection of $\mathcal{C}$ with the halfspace $x\leqslant 0$.
See Figure~\ref{fig:domain}(a). Since we took $r$ in the interior of $\mathcal{C}$, both $\mathcal{C}'$ and $\mathcal{C}''$
have nonempty interiors.

Let $\mathcal{C}_1,\dots ,\mathcal{C}_n$ be the cones after the above transformation. 
For each $i$, let $(a_i,b_i,c_i)$ be the apex of $\mathcal{C}_i$.
Recall that, by assumption, each cone $\mathcal{C}_i$ is the translation of $\mathcal{C}$
that brings $(0,0,0)$ to $(a_i,b_i,c_i)$.
We split each cone $\mathcal{C}_i$ into two cones, denoted $\mathcal{C}'_i$ and $\mathcal{C}''_i$,
using the plane $x=a_i$. 
Notice that $\mathcal{C}'_1,\dots,\mathcal{C}'_n$ are translations of $\mathcal{C}'$,
and $\mathcal{C}''_1,\dots,\mathcal{C}''_n$ are translations of $\mathcal{C}''$.
We split the problem into two subproblems, in one of them want to find a point-cone incidence
between $P$ and $\mathcal{C}'_1,\dots,\mathcal{C}'_n$, and in the other we want to
find a point-cone incidence between $P$ and $\mathcal{C}''_1,\dots,\mathcal{C}''_n$.
Any point-cone incidence $(p,\mathcal{C}'_i)$ or $(p,\mathcal{C}''_i)$ corresponds
to a point-cone incidence $(p,\mathcal{C}_i)$, and vice versa.
We explain how to solve the point cone incidence for $P$ and cones $\mathcal{C}'_i$; 
the incidence for cones $\mathcal{C}''_i$ is similar. 

Recall that the origin is the apex of $\mathcal{C}'$. 
We define $M$ to be the polygon obtained by intersecting
$\mathcal{C}'$ with the horizontal plane $z=1$. 
Note that $M$ lies on the halfspace $x\ge 0$.
Our choice of $r$ in the interior of $\mathcal{C}$ implies 
that $M$ has a nonempty interior and $(0,0,1)$ lies on the (relative) interior of a boundary edge of $M$. 
See Figure~\ref{fig:domain}(a). Since $M$ is a convex polygon that is the intersection of $n$ halfplanes with the  plane $z=1$, it can be computed in $O(n\log n)$ time.
In the rest of description, we consider $M$ being in $\mathbb{R}^2$, where we just drop the $z$-coordinate as in Figure~\ref{fig:domain}(b).

\begin{figure}[tb]
	\centering
	\setlength{\tabcolsep}{0in}
	$\begin{tabular}{cc}
	\multicolumn{1}{m{.62\columnwidth}}{\centering\includegraphics[width=.61\columnwidth]{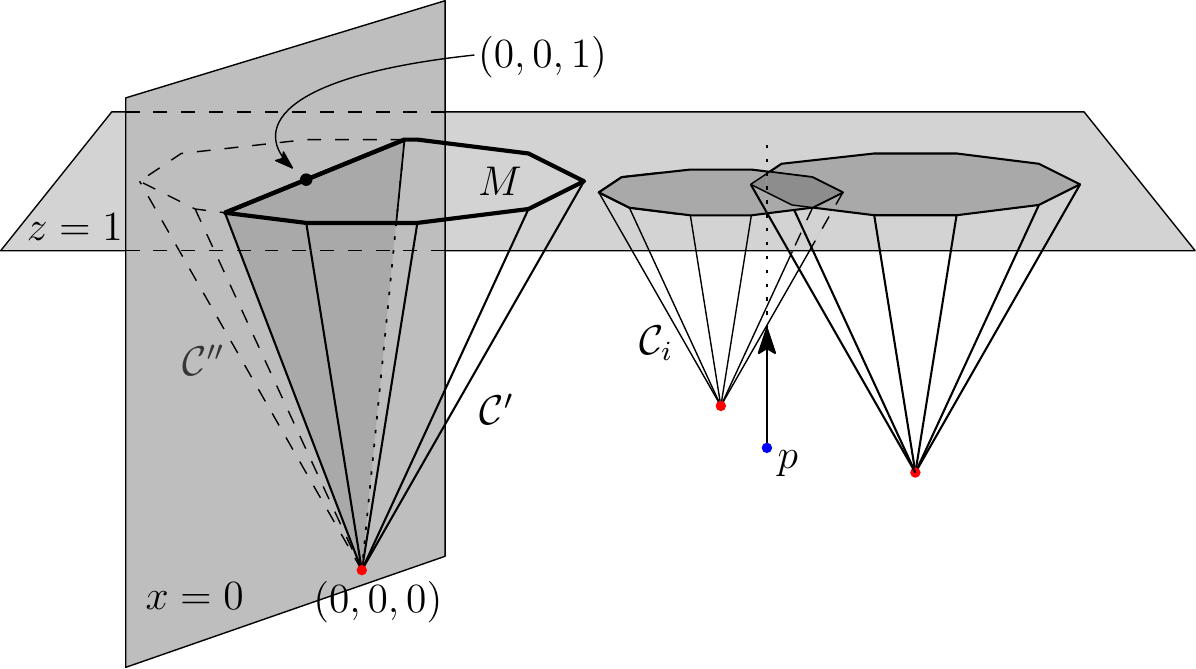}}
	&\multicolumn{1}{m{.38\columnwidth}}{\centering\includegraphics[width=.37\columnwidth]{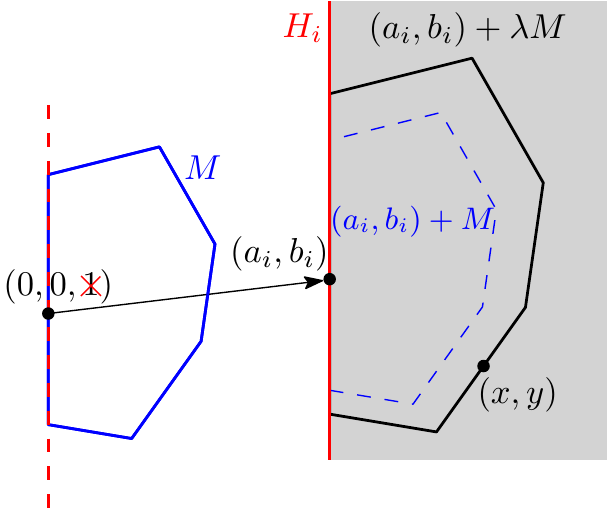}}\\
	(a)&(b)
	\end{tabular}$
	\caption{(a) The cone $\mathcal{C}$ that is splitted to $\mathcal{C}'$ and $\mathcal{C}''$, and the polygon $M$ which is the intersection of the plane $z=1$ with $\mathcal{C}'$. (b) The domain $H_i$, and the translation	of $M$ that brings $(0,0)$ to $(a_i,b_i)$ followed by a scale with factor $\lambda$.}
	\label{fig:domain}
\end{figure}

Let $H_i$ denote the projection of $\mathcal{C}'_i$ on the $xy$-plane. Note that $H_i$ is the halfplane defined by $x\ge a_i$ because we took $r$ 
in the interior of $\mathcal{C}$. 
See Figure~\ref{fig:domain}.
Let $H$ be the union of all halfplanes $H_i$, and note that $H$ is defined by $x\geqslant \min \{ a_1,\dots, a_n\}$.

The boundary of every cone $\mathcal{C}'_i$ can be interpreted as a function $f_i\colon H_i\rightarrow \mathbb{R}$ where $f_i(x,y)= \min \{ \lambda \in \mathbb{R}_{\ge 0} \mid (x,y,\lambda)\in \mathcal{C}'_i \}$. 
Alternatively, for every $(x,y)\in H_i\setminus \{ (a_i,b_i)\}$ we have 
\[
f_i(x,y)=c_i+\min \{\lambda\geqslant 0\mid (x,y)\in (a_i,b_i)+\lambda M\} = 
c_i+\min \{\lambda>0\mid \frac{(x,y)-(a_i,b_i)}{\lambda} \in M\},
\]
where $\lambda$ is the smallest amount that $M$ must be scaled, after a translation to $(a_i,b_i)$, to include $(x,y)$; see Figure~\ref{fig:domain}(b).
Note that if $\mathcal{C}'_i$ contains a point $(x,y,z)$, it also contains $(x,y,z')$ for all $z'\geqslant z$.
Therefore, the surface $\{ (x,y,f_i(x,y)\mid (x,y)\in H_i \}$ precisely defines
the boundary of $\mathcal{C}'_i$. Based on this, to decide whether a point $(x,y,z)$ lies in $\mathcal{C}'_i$, 
it suffices to check whether $(x,y)\in H_i$ and $z\geqslant f_i(x,y)$. Notice that every $f_i$ is a convex function on domain $H_i$.
We extend the domain of each $f_i$ to $H$ by setting $f(x,y)=\infty$ 
for all $(x,y)\in H\setminus H_i$. In this way all the functions $f_i$ are defined 
in the same domain $H$.

Let us denote by $F$ the family of functions $\{ f_1,\dots,f_n \}$,
and define the pointwise minimization function $f_{\min}(x,y)=\min \{ f_1(x,y),\dots,f_n(x,y) \}$ for every $(x,y)\in\RR^2$.
For simplicity, we assume that the surfaces defined by $F$ are in \emph{general position} in the sense that, before extending the domains of the functions $f_i$, the following conditions hold:
(i) no apex of a cone lies on the boundary of another cone, that is, 
{$f_j(a_i,b_i)\neq c_i$ for all $i\neq j$,
(ii) any three surfaces defined by $F$ have a finite number of points in common, and
(iii) no four surfaces defined by $F$ have a common point. 
Such assumptions can be enforced using infinitesimal perturbations.

For each $i$, let $\mathcal{R}_i$ be the subset of $H$ where $f_i$ gives the minimum among all functions in $F$, that is
\[
\mathcal{R}_i = \{ (x,y)\in H\mid f_i(x,y)=f_{\min}(x,y)\}.
\]
This introduces a partition the plane into regions $\mathcal{R}_i$; we will refer to this partition by {\em minimization diagram}. We note that $\mathcal{R}_i$ can be empty; this occurs when the apex $(a_i,b_i,c_i)$ of $\mathcal{C}_i$ is contained in 
the interior of some cone $\mathcal{C}'_j$, with $j\neq i$.
We show some folklore properties of the regions $\{\mathcal{R}_1,\dots,\mathcal{R}_n\}$. In our following description, $[n]$ denotes the set $\{1,2,\dots, n\}$.

\begin{lemma}
	\label{le:starshaped}
	For any index $i\in [n]$, the region $\mathcal{R}_i$ is star-shaped with respect to the point $(a_i,b_i)$.
	For any three distinct indices $i,j,k\in [n]$, the intersection $\mathcal{R}_i\cap \mathcal{R}_j\cap \mathcal{R}_k$
	contains at most two points. 
	For any two distinct indices $i,j\in [n]$ and every line $\ell$ in $\mathbb{R}^2$ where $(a_i, b_i)$ and $(a_j,b_j)$ lie on the same side of $\ell$, the intersection $\ell\cap \mathcal{R}_i\cap \mathcal{R}_j$ contains at most two points. 
\end{lemma}
\begin{proof}
	To verify $\mathcal{R}_i$ being star-shaped with respect to $(a_i,b_i)$, 
	the proof of~\cite[Corollary~2.5]{McAllister1996} applies.
	To verify the second claim, notice that if $\mathcal{R}_i\cap \mathcal{R}_j\cap \mathcal{R}_k$ have three or more points, then by connecting those points to $(a_i,b_i)$, $(a_j,b_j)$ and $(a_k,b_k)$ with line segments, we would get a planar drawing of the graph $K_{3,3}$, which is impossible. To verify the third claim, note that if for some line $\ell $ the intersection $\ell\cap \mathcal{R}_i\cap \mathcal{R}_j$ have three or more points, then again we would get an impossible planar drawing of $K_{3,3}$ as follows: we connect those points to $(a_i,b_i)$, $(a_j,b_j)$, and to an arbitrary point to the side of $\ell$ that does not contain $(a_i,b_i)$ and $(a_j,b_j)$. 
\end{proof}

\begin{lemma}
	\label{le:basic-operations}
	After $O(n \log n)$ preprocessing time on $M$ we can solve the following problems in $O(\log n)$ time:
	\begin{itemize}
		\item Given a point $p$ and an index $i\in [n]$, decide whether or not $p\in \mathcal{C}_i$.
		\item Given three distinct indices $i,j,k\in [n]$, compute $\mathcal{R}_i\cap \mathcal{R}_j\cap \mathcal{R}_k$.
		\item Given two distinct indices $i,j\in [n]$ and a vertical line $\ell$ in $\mathbb{R}^2$,
		compute $\ell\cap \mathcal{R}_i\cap \mathcal{R}_j$.
	\end{itemize}
\end{lemma}
\begin{proof}
	We compute $M$ explicitly in $O(n \log n)$ time
	and store its vertices and edges cyclically ordered in an array. Let $-M$ denote $\{(-x,-y)\mid (x,y)\in M\}$.
	For each vertex $v$ of $M$ we choose an outer normal $\overrightarrow{n_v}$ vector.
	We also store for each vertex $v$ of $M$  
	the vertex of $-M$ that is extremal in the direction $\overrightarrow{n_v}$.
	Having $M$, this can be done in linear time by walking
	through the boundaries of $M$ and $-M$ simultaneously.
	This finishes the preprocessing.
	
	For the first claim, we are given an index $i$ and a point $p=(p_x,p_y,p_z)$. By performing binary search on the edges of $M$ we can find the edge that intersects the ray with direction $(p_x,p_y)-(a_i,b_i)$ in $O(\log n)$ time. This edge determines
	the value $\min \{\lambda\ge 0\mid (p_x,p_y)\in (a_i,b_i)+\lambda M\}$, which in turn gives 
	$f_i(p_x,p_y)$. By comparing $f_i(p_x,p_y)$ with $p_z$ we can decide whether or not $p\in \mathcal{C}_i$.
	
	Now we prove the second claim. By Lemma~\ref{le:starshaped}, $\mathcal{R}_i\cap \mathcal{R}_j\cap \mathcal{R}_k$
	contains at most two points. Assume, without loss of generality,
	that $i=1$, $j=2$, $k=3$ and $c_1=\max \{ c_1,c_2,c_3\}$.
	Let $P_1$ be the (degenerate) polygon with a single vertex $(a_1,b_1)$.
	Let $P_2$ and $P_3$ be the convex polygons $(a_2,b_2)+(c_1-c_2)M$ and $(a_3,b_3)+(c_1-c_3)M$, respectively. The polygons $P_2$ and $P_3$ might also be single points if $c_2=c_1$ and $c_3=c_1$.
	A point $(x,y)$ belongs to $\mathcal{R}_1\cap \mathcal{R}_2\cap \mathcal{R}_3$
	if and only if for some $\lambda$ the polygon $(x,y)+\lambda(-M)$ is tangent
	to $P_1$, $P_2$ and $P_3$. 
	
	If there is some containment between the polygons $P_1$, $P_2$ and $P_3$, i.e., one polygon is totally contained in other polygon, 
	then $\mathcal{R}_i\cap \mathcal{R}_j\cap \mathcal{R}_k$ is empty. Assume that there is no containment between these polygons. Now we are going to find two convex polygons $P'_2$ and $P'_3$ such that $P_1$, $P'_2$, $P'_3$ are pairwise interior disjoint. 
	If $P_2$ and $P_3$ are interior disjoint, we take $P'_2=P_2$ and $P'_3=P_3$.
	Otherwise, the boundaries of $P_2$ and $P_3$ intersect at most twice because
	they are convex polygons. We compute the intersections $q_1,q_2$ between the boundary
	of $P_2$ and $P_3$ in $O(\log n)$ time.
	We use the the segment $q_1q_2$ to cut $P_2\cap P_3$ so that we obtain
	two interior disjoint convex polygons $P'_2\subset P_2$ and $P'_3\subset P_3$
	with $P'_2\cup P'_3= P_2\cup P_3$.
	The polygon $P'_2$ is described implicitly by the segment $q_1q_2$
	and the interval of indices of $M$ that describe the portion of $P_2$ between $q_1$ and $q_2$.
	The description of $P'_3$ is similar. With this description, we can perform 
	binary search on the boundaries of $P'_2$ and $P'_3$.
	
	Now we want to find the (at most two) scaled copies of $-M$ that can be translated
	to touch $P_1$, $P'_2$ and $P'_3$. Since the polygons are disjoint, we can use
	the tentative prune-and-search technique of 
	Kirkpatrick and Snoeyink~\cite{Kirkpatrick1995}
	as used in~\cite[Lemma~3.15]{McAllister1996}. The procedure makes $O(\log n)$ steps,
	where in each step we locate the extreme point of $-M$ in the direction $\overrightarrow{n_v}$ for some vertex $v$ of $P'_i$. Since such
	vertices are precomputed, we spend $O(1)$ time in each of
	the $O(\log n)$ steps used by the tentative prune-and-search.\footnote{
		The running time in~\cite[Lemma~3.15]{McAllister1996} has an extra logarithmic factor
		because they spend $O(\log m)$ time to find the extremal vertex in a polygon $M$
		with $m$ vertices.} 
	
	The proof of the third claim is similar to that of previous claim, where we treat $\ell$ as a degenerate polygon.
\end{proof}

Let $A$ be the set of points $\{ (a_1,b_1),\dots (a_n,b_n)\}$ defined by the apices of the cones.
We use the sweep-line algorithm of~\cite{McAllister1996} to compute a representation of the
minimization diagram. More precisely, we sweep $H$ 
with a vertical line $\ell\equiv \{ (x,y)\mid x=t\}$, where $t$ goes from $-\infty$ to $+\infty$. 
In our case, the sweep line and the sweep front (the beach line) are the same because
future points of $A$ do not affect the current minimization diagram.
During the sweep, we maintain (in a binary search tree) the intersection
of $\ell$ with the regions $\mathcal{R}_i$, sorted as they occur along the line $\ell$, possibly with repetitions.

There are two types of events. A {\em vertex event} (or {\em circle event}) 
occurs when the sweep front goes over 
a point of $\mathcal{R}_i\cap \mathcal{R}_j\cap \mathcal{R}_k$. In our case, this is when $\ell$
goes over such a point. A \emph{site event} occurs when the sweep line $\ell$ (and thus the sweep front)
goes over a point of $A$.
The total number of these events is linear.

Vertex events will be handled in the same way as in McAllister \etal \cite{McAllister1996}. We describe how to handle site events.
At a site event, we locate the point $(a_i,b_i)$ in the current region $\mathcal{R}_j$
that contains it, as $\mathcal{R}_i$ could be empty. 
As shown in~\cite[Section~3.2]{McAllister1996}, this location can be done in $O(\log n)$ time
using auxiliary information that is carried over during the sweep.
Once we have located $(a_i,b_i)$ in $\mathcal{R}_j$, we compare $f_j(a_i,b_i)$ with $c_i$
to decide whether or not $(a_i,b_i,c_i)$ is contained in $\mathcal{C}'_j$.
If $(a_i,b_i,c_i)$ belongs to $\mathcal{C}'_j$, with $j\neq i$, then the region $\mathcal{R}_i$ is empty,
and we can just ignore the existence of $\mathcal{C}_i$. Otherwise, $\mathcal{R}_i$ is not empty, and we have
to insert it into the minimization diagram and update the information associated
to $\ell$. Overall, we spend $O(\log n)$ time
per site event $(a_i,b_i)$, plus the time needed to find future vertex events triggered by the current 
site event.

Whenever the line $\ell$ passes through a point $(p_x,p_y)$, where $(p_x,p_y,p_z)\in P$,
we can apply the same binary search on the sweep line as for site events.
This means that in time $O(\log n)$ we locate the region $\mathcal{R}_j$ that contains $(p_x,p_y)$.
Then we check whether or not $(p_x,p_y,p_z)$ belongs to $\mathcal{C}'_j$; this would take an additional $O(\log n)$ time by the first claim in Lemma~\ref{le:basic-operations}.
Therefore, we can decide in $O(\log n)$ time whether or not the point $(p_x,p_y,p_z)\in P$ belongs
to any of the cones $\mathcal{C}'_1,\dots,\mathcal{C}'_n$.

At any event (site or vertex event) that changes the sequence of regions $\mathcal{R}_i$
intersected by $\ell$, we have to compute possible new vertex events. By Lemma~\ref{le:basic-operations}, this computation takes $O(\log n)$ time; the third claim in this lemma takes care of so-called vertices at infinity in~\cite{McAllister1996}. We note that in~\cite{McAllister1996} this step takes $O(\log n \log m)$ time
because of their general setting. 

To summarize, we have a linear number of events each taking $O(\log n)$ time. Therefore, we can decide the existence of a point-cone incidence in $O(n\log n)$ time.
This finishes the proof of Theorem~\ref{point-cone-thr}

\paragraph{Acknowledgement.}
This work initiated at the {\em Sixth Annual Workshop on Geometry and Graphs}, March 11-16, 2018, at the Bellairs Research Institute of McGill University, Barbados. The authors are grateful to the organizers and to the participants of this workshop. Segrio Cabello was supported by the Slovenian Research Agency, program P1-0297 
and projects J1-8130, J1-8155.

Ahmad Biniaz was supported by NSERC Postdoctoral Fellowship. Sergio Cabello was supported by the Slovenian Research Agency, program P1-0297 
and projects J1-8130, J1-8155. Paz Carmi was supported by grant 2016116 from the United States – Israel Binational Science Foundation. Jean-Lou De Carufel, Anil Maheshwari, and Michiel Smid were supported by NSERC. Saeed Mehrabi was supported by NSERC and by Carleton-Fields Postdoctoral Fellowship.
	
\bibliographystyle{abbrv}
\bibliography{Consistent-Subset}

\end{document}